\documentclass[12pt,reqno]{amsart}

\usepackage{amssymb}
\usepackage{amsfonts}
\usepackage{amsmath}
\usepackage{latexsym}
\usepackage{amscd}

\usepackage{diagrams}

\newtheorem{theorem}{Theorem}[section]
\newtheorem{proposition}[theorem]{Proposition}
\newtheorem{corollary}[theorem]{Corollary}
\newtheorem{lemma}[theorem]{Lemma}

\newtheorem{remark}[theorem]{Remark}
\newtheorem{definition}[theorem]{Definition}
\newtheorem{Problem}[theorem]{Problem}

\newenvironment{Remark}{\begin{remark}\begin{em}}{\end{em}\end{remark}}

\newenvironment{problem}{\begin{Problem}\begin{em}}{\end{em}\end{Problem}}

\DeclareMathOperator{\A}{\mathcal{A}}

\DeclareMathOperator{\Sym}{\text{Sym}}
\DeclareMathOperator{\Sp}{\text{Sym}^{\geq 0}}
\DeclareMathOperator{\Spp}{\text{Sym}^{>0}}
\DeclareMathOperator{\gyr}{gyr}

\def\R{{\mathbb R}}

\def\M{\mathcal{M}}
\def\a{\mathbf{a}}
\def\b{\mathbf{b}}
\def\c{\mathbf{c}}

\def\x{\mathbf{x}}
\def\y{\mathbf{y}}
\def\v{\mathbf{v}}
\def\u{\mathbf{u}}
\def\w{\mathbf{w}}
\def\oa{\overrightarrow}
\def\p{\frak{p}}
\def\la{\langle}
\def\ra{\rangle}
\address{Department of Mathematics, Louisiana State University,
Baton Rouge, LA70803, USA} \email{lawson@math.lsu.edu}

\begin{document}
\title[Einstein Velocity and Gyrogroups]{Special Relativity, Einstein Velocity Addition, and Gyrogroups:  An Introduction}
\author[Lawson]{Jimmie Lawson}

\maketitle
\begin{abstract}
In these notes we give an  introductory unified treatment to the topics of special relativity, Lorentz transformations and 
the Lorentz group, Einstein velocitiy addition, and gyrogroups and gyrovector spaces.   An effort has been made to
present the material in a manner that is accessible to non-specialists and graduate students, and may even serve
as the basis for a graduate course or seminar.
\end{abstract}

\bigskip
The material for this article had its origin in a graduate seminar taught by the author as a first introduction
to the mathematics of special relativity with a particular focus on Einstein velocity addition and its encoding
in the language and structure of gyrogroups and gyrovector spaces. The favorable reception of this material
encouraged the author to make this material and its shortened and simplified presentation more widely 
available.  Most of of the following material can be found in much greater detail and depth in A.\ A.\ Ungar's
monograph \cite{Ung}. A somewhat variant, but overlapping, approach can be found in  Chapter 1 of 
Y.\ Friedman's monograph \cite{Fri}.  It is the hope of the author, however, that what follows might be
more suitable for a first look at the material or for an introductory seminar.

\section{Introduction to Special Relativity}  
A common understanding among the ancients of physical dynamics was that 
objects near the earth left to themselves would move themselves as close to the center of the earth
as possible.  Heavenly objects, on the other hand, were perfectly formed objects that would move in
perfect circles around the earth.  Close observation, however, revealed that this was not true for the
planets, so Ptolemy used orbits described by epicycles (paths obtained from a 
circular motion around a center, which is also moving in a circular motion) and eccentric circles (with the
earth not at the center) to model the movement of heavenly bodies, in particular the planets.

\subsection{The Principle of Inertia} Such ideas persisted until the time of Galileo.  
His experiments with objects rolling down ramps
and other physical and mental experiments led him to the conclusion that objects free from external
influence would either remain at rest or move in a straight line at a constant speed.  This is 
sometimes known as Galileo's Principle of Inertia and was popularized as Newton's first law.  
We recall from the vector geometry of $\R^3$ that an object moving in a straight line at constant
speed, or equivalently moving at some constant velocity $\v\in\R^3$, 
has a parametric description of the form $\x(t)=\x_0+t\v $ (the solution of $\dot \x(t)=\v$, $\x(0)=\x_0$),  where 
we may view $t$ as the time parameter.  Such motion is called \emph{rectilinear motion}
(or \emph{linear motion}).

\subsection{Galilean relativity} Further, Galileo and Newton assumed a principle of relativity for dynamics: 
all observers in a system at rest or moving
with a constant velocity will encounter the same laws of dynamics (the physics of objects
or masses and their movements).  For example, a tennis player would experience no difference playing
tennis on land or playing in the depths of a large ocean liner sailing on a smooth sea at a constant
velocity.

More formally, we can define an \emph{inertial frame}, \emph{reference frame}, or simply \emph{frame}
as a coordinatization of space-time $\R\times \R^3$ in which rectilinear motion has the 
description $\x(t)=\x_0+t\v$ for $t\in\R$ and $\x_0,\v\in\R^3$.  We alternatively say 
that Newton's first law of motion, i.e., Galileo's Principle of Inertia, holds.  Suppose that
the coordinates of a frame $S'$ can be computed from those of $S$ by
a \emph{basic Galilean transformation}:
\begin{eqnarray}\label{E1.1}
\x'=\x-t\v \mbox{ and } t'=t,
\end{eqnarray}
where $\x\in\R^3$, $t\in\R$ are the coordinates in the first frame $S$ and $\x'\in\R^3$, $t'\in \R$ are coordinates in the 
second $S'$.  Then we say that the inertial frame $S'$ is moving away from the frame $S$ at the (constant) velocity
$\v\in\R^3$. Note that the two coordinate systems agree at time $t=0$.

\begin{problem}
Find the Galilean transformation if a second  inertial frame $S'$ is moving away from the first $S$ at velocity $\v$
and the second one has coordinate $\x_0$ for the origin of the first at time $t=0$.
\end{problem}

\begin{problem}
Show that linear motion is preserved by basic Galilean transformations, and hence one of the coordinate systems is an
inertial frame iff the other one is.
\end{problem}

\begin{problem}  
(i) Suppose $S'$ is computed from $S$ by the Galilean transformation $\x'=\x-t\v $, $t'=t$, and that
an object is moving with constant velocity $\mathbf{w}$ with respect to the frame $S'$.  Find its velocity with respect to $S$. \\
(ii) Suppose $S'$  is moving away from $S$ at constant velocity $\v$ and $S''$ is moving away from $S'$ at constant
velocity $\mathbf{w}$.  How fast is $S''$ moving away from $S$?
\end{problem}

\begin{problem}
Show that the inverse of a basic Galilean transformation is a Galilean transformation and that the composition of two is another.
\end{problem}

\begin{problem}
Show that a basic Galilean transformation is a linear map from $\R^3\oplus\R\approx \R^4$ to itself.
\end{problem}

\begin{problem}
(i) Show for a basic Galilean transformation that $\x_2'-\x_1'=\x_2-\x_1$ (more precisely, this means that if the transformation carries
$(\x_1,t)$ to $(\x_1',t')$ and $(\x_2,t)$ to $(\x_2',t')$, then the asserted equality holds).
(ii) Think about the previous result long enough for it to be intuitively obvious 
(but you may quit in 24 hours).\\
(iii) Formulate a version of (i) for two frames $S'$ and $S''$ moving away from a fixed frame $S$
at velocities $\v'$ and $\v''$ resp. 
\end{problem}

The next problem illustrates how the quantities in the Newtonian law of motion $F=ma$ are unaltered by
Galilean transformations.

\begin{problem}
Consider two point masses in space of mass $m_1$ and $m_2$ 
and suppose that the force between them depends only on their separation:
\begin{eqnarray}
F(\x_1,\x_2)= f(\x_1-\x_2)
\end{eqnarray}
when mass $m_1$ is at $\x_1\in\R^3$ and mass $m_2$ is at $\x_2\in\R^3$.\\
(i) Show that for a basic Galilean transformation from $S$ to $S'$, 
$F(x_1',x_2')=F(x_1,x_2)$.  \\
(ii) From Newton's second law $F=ma$, conclude that $a_2=a_2'$, where $a_2$ resp.
$a_2'$ is the acceleration of $m_2$ in the inertial frame $S$ resp.\  $S'$. (In Newtonian
dynamics the inertial mass $m$ is a constant under Galilean transformations.)

\end{problem}

\subsection{The Einstein postulates}
In the nineteenth century physicists postulated the existence of an ``ether'' in space that would 
enable the transmission of electromagnetic and light waves in space (based on their experience with
sound waves, water waves, vibrating strings, etc.).  However, all efforts to
experimentally verify this ether, such as the famous Michaelson-Morley experiment, failed completely
to detect any such ether.  This led physicists to consider further explanations and culminated in the
original and penetrating insights of Albert Einstein,  who introduced his theory of special relativity in 1905.
Einstein based his theory on two basic postulates:\\
\textbf{Postulate 1}:  \textit{All inertial frames are equivalent with respect to all the laws of physics.}\\
\textbf{Postulate 2}: \textit{The speed of light in empty space has the same value $c$ in any inertial frame.}

The first postulate expanded on the relativity principle of Galileo and Newton by assuming its validity
for all laws of physics, in particular those of electricity and magnetism (Maxwell's equations were not
preserved by general Galilean transformations).  It is a tribute to the insight
of Einstein that a whole new dynamics can be built upon these two brief statements.

One of the startling consequences of Einstein's postulates was the inference that time was relative,
i.e., dependent on one's frame of reference.  Newton has assumed a universal time: ``Absolute, true,
and mathematical time, of itself, and from its own nature, flows equably without relation to anything 
external.''  Einstein recognized that judgements about time were inextricably tied up with judgments
about simultaneity, but defining the latter was problematic if information can only be transmitted at 
a finite speed.  

Einstein used the large but finite speed $c$ of electromagnetic signals (in a vacuum) for relating time 
measurements.  For
example, he introduced the following method of synchronizing clocks at different locations in a given
inertial frame.  If a signal starts at a point $A$ at time $t=t_0$, is reflected back by a mirror at $B$, and
returns to $A$ at time $t_1$, then the time at which the signal reached $B$ is defined as 
$t_0+(1/2)\Delta(t)$, where $\Delta(t):=t_1-t_0$. In this manner clocks at all points can be 
synchronized in the given frame. 

\subsection{Spacetime}
In special (and general) relativity time is relative in the sense that it depends on the inertial frame we choose
and varies from frame to frame.  We thus no longer refer to ``space-time'', but to ``spacetime'' since the two
are inextricably bound together.  Mathematically we may think of spacetime as an ``uncoordinatized'' 
four-dimensional space $M$, a manifold.  Members of spacetime $M$ are called \emph{events}.
We assume that we have a notion of rectilinear motion in $M$, and that there exist \emph{inertial frames}, 
bijections $S$ from $\R\times \R^3$ to $M$ that 
endow $M$ with a coordinate system such that rectilinear motion in $M$ corresponds to rectilinear 
motion in the frame, i.e., has the form $t\mapsto
(t,\x_0+t\v)$ in the coordinates of the frame $S$.  In the coordinates of any reference frame, the first coordinate 
of the frame is the \emph{time coordinate} and the last
three are called the \emph{space coordinates}.  We assume that an unhindered light signal, 
or more generally an electromagnetic signal,
moves rectilinearly with speed $c$ in any inertial frame and that this is the maximal attainable
speed.   A \emph{world line} is a mapping from $\R$ into $M$ which is given by 
$t\mapsto (t,\x(t))$ in the coordinates of a frame, where $\x(\cdot):\R\to\R^3$ is a continuous, piecewise
smooth path with speed bounded by $c$.

\begin{problem} \label{P:Spt} Suppose that $\alpha(t)=(t,\x_0+ t\v)$ describes rectilinear motion in $M$ for some
frame $S$.  Argue that $\alpha((s+t)/2)$ is the midpoint of $\alpha(t)$ and $\alpha(s)$ in
$\R\oplus\R^3$.
\end{problem}

\begin{problem}  Argue that a coordinate frame composed with the inverse of another coordinate frame
preserves rectilinear motion.
\end{problem}

\begin{problem}  Show that if two coordinate systems have the same spacetime origin, then the 
map in the preceding problem is a linear one.  (Hint: Use Problem \ref{P:Spt} to show the map 
is additive, i.e., preserves vector addition.)
\end{problem}

\subsection{Minkowski diagrams}
Given two inertial frames moving at constant velocity with respect to each other,  we may conveniently
recoordinatize the frames by taking as the $x$-axis the direction of the velocity and assuming the 
coordinate systems agree at time $t=0$.  The Galilean transformation between the systems then 
simplifies to 
\begin{equation}
x'=x-vt,~~~y'=y,~~~z'=z,~~~t'=t.
\end{equation}
To illustrate graphically various features of special relativity for the two inertial frames, 
we suppress the $y$ and $z$ coordinates and draw two-dimensional space-time diagrams
with the $x$-axis horizontal and the $t$-axis, scaled by a factor of $c$, vertical.  We assume 
that the $x,t$-coordinates represent the first coordinate system, which is at rest with respect 
to the second one, which is moving at a velocity $v$ in the $x$-direction.

We draw world lines on the diagram that display the complete history of a one-dimensional motion.  
We assume that time is synchronized at all points with respect to the given frame, and then the
world line of a moving object consists of all positions it occupies together with the time at that position.
For example, objects at rest within the 
given frame of reference have world lines that are vertical lines with the constant 
$x$-coordinate being their fixed position.  If we take any
fixed point $B$ at $x_B$ on the $x$-axis and send out a light signal in the positive $x$-direction at time $t=0$,
then the world line of the signal is a line of slope $1$ emanating from $x_B$ at $t=0$.

\begin{problem}
Suppose that an object is moving with constant velocity $\v$ along the $x$-axis.  Show that its
world line has slope $c/v\geq 1$.  
\end{problem}

\subsection{Simultaneity}
Consider three observation stations $A<B<C$ equally spaced along the $x$-axis in a Minkowski
diagram for an inertial frame $S$ at points $x_A,x_B,x_C$.  
Assume first that $A,B,C$ are at rest in this frame.  Then their
world lines are vertical.  Suppose that a light signal is sent out from $B$ at time $t=0$, traveling at
speed $c$ both forward and backward along the $x$-axis.  
\begin{problem}\label{P1.9}
(i) Argue that if the clocks at $A$ and $C$ are synchronized in $S$, then the two signals must 
strike $A$ and $C$ at the same time.\\
(ii) Show that (i) is equivalent to the point of intersection of the world line for $A$ and the world line
for the signal moving to the left having the the same $t$-coordinate as the point of intersection of the world
line for  $C$ and the world line for the signal moving to the right.  Illustrate with appropriate
world lines in a Minkowski diagram.
\end{problem}
In light of (ii) of the preceding problem, we see that a line of simultaneous events is a line
for which $t$ is constant, i.e., a horizontal line or line parallel to the $x$-axis.  

Now suppose that $A$, $B$, and $C$ are all moving with speed $v$ along the $x$-axis,
i.e., are at rest in the inertial frame $S'$ that is moving with respect to $S$ at a speed $v$
along the $x$-direction.  Again a signal is sent from $B$ at time $t=0=t'$, and by the preceding
reasoning its reaching of $A$ and $C$ must be simultaneous events in the frame $S'$.
However in $S$ the signal travels further to the receding $C$ than to the approaching $A$,
and hence the two events are no longer simultaneous in $S$.  Thus simultaneity of events 
depends on the frame of reference, if we demand that the speed of light have the same
value $c$ in all reference frames.
\begin{problem}
(i) Draw the world lines for $A$, $B$, $C$, and the signals in the reference frame $S$.
A line of simultaneity for $S'$ is given by the (oblique) line connecting the point of intersection of 
the world line for $A$ and the world line for the signal moving to the left and  the point of 
intersection of the world line for $C$ and the world line for the signal moving to the right.\\
(ii) The axes for the moving frame $S'$ with respect to the stationary frame $S$ can be found
by taking the $t'$-axis to be the line through the origin parallel to the world lines of $A$, $B$,
and $C$, and the $x'$-axis the line through the origin parallel to any line of simultaneity.
Draw these lines in the Minkowski graph.
\end{problem}

\section{Lorentz transformations}
Suppose that one observer uses an inertial coordinate frame $S$ with coordinates $(ct,x,y,z)$,
and another observer uses another inertial coordinate frame $S'$ given by coordinates
$(ct',x',y',z')$. We assume that the frames are in motion with respect to each other 
so that the $S$-frame observer 
observes the other to be moving at speed $v$ in the $+x$-direction and the $S'$-frame observer 
observes the other to be moving at the same speed in the $-x'$-direction, with the $y,z$-coordinates
being the same for both observers.  We further assume that the space-time origin agrees for
the two frames. Let us suppose 
we know the location of an event according to one observer's coordinate frame and wish to 
determine the location according to the other coordinate frame. We transform the coordinates 
of the event from the one to the other by doing a Lorentz coordinate transformation. A reasonable
guess for the Lorentz coordinate transformation equations is to do a slight generalization of
the Galilean case, namely multiply by a constant $\gamma$, at least for the space coordinates:
\begin{eqnarray}\label{E2.1}
x' &=& \gamma(x- vt),~~~y'=y,~~~z'=z;\nonumber \\
x &=& \gamma(x'+ vt'),~~~y=y',~~~z=z'.
\end{eqnarray}
Note by the first equation that the world line of the origin, $x'=0$, in $S'$ transforms in $S$ 
to $x=vt$ and conversely $x=0$ transforms to $x'=-vt$, justifying the choice of $v$ for the 
$t$-coefficient.

If a light signal originates from the origin in the $x$-direction, then it has world line $x=ct$ in $S$ and
$x'=ct'$ in $S'$ by Einstein's Postulate 2.  Substituting these values into  (\ref{E2.1})
yields for any point on the world line of the signal
\begin{eqnarray*}
ct' &=& \gamma(ct-vt)=\gamma(c-v)t\\
ct &=& \gamma(ct'+vt')=\gamma(c+v)t';
\end{eqnarray*}
Transposing the second equation, then dividing the two equations and cross-multiplying
yields
\begin{equation*}
c^2=\gamma^2(c^2-v^2)
\end{equation*}
Solving for $\gamma$ yields
\begin{equation}\label{E2.2}
\gamma_v:=\gamma=\frac{1}{(1-v^2/c^2)^{1/2}} \mbox{ for } \vert v\vert<c.
\end{equation}
The constant $\gamma_v$ is sometimes called the \emph{Lorentz factor}.

\begin{problem}
Solve the system (\ref{E2.1}) to obtain
\begin{eqnarray}\label{E2.3}
ct &=& \gamma(ct' + vx'/c) \nonumber \\
ct' &=& \gamma(ct-vx/c)
\end{eqnarray}
\end{problem}

Combining the previous results, we can express the Lorentz transformations by
\begin{eqnarray}
t'&=& \gamma(t-vx/c^2) \nonumber\\
x' &=& \gamma(x- vt),~~~~y'=y,~~~~z'=z.
\end{eqnarray}
and
\begin{eqnarray}
t&=& \gamma(t'+vx'/c^2) \nonumber\\
x &=& \gamma(x'+ vt'),~~~~y=y',~~~~z=z'.
\end{eqnarray}
The previous Lorentz transformations are called a \emph{Lorentz boost} 
and often written in matrix form, for example:
\begin{equation}\label{E:boost1}
\left[\begin{matrix}ct'\\ x'\\ y'\\ z'\end{matrix}\right]=
\left[\begin{matrix}\gamma &-\gamma v/c& 0 & 0\\
-\gamma v/c &\gamma & 0 & 0\\
0 & 0 & 1 & 0\\ 0 & 0 & 0& 1
\end{matrix}\right]
\left[\begin{matrix}ct\\ x\\ y\\ z\end{matrix}\right]
\end{equation}

\begin{problem}
 Show that the limit as $v/c\to 0$ of the Lorentz transformation (2.7)
reduces to a Galilean transformation.  
\end{problem}

\begin{problem} A frame $S'$ has velocity $0.6c$ in the $x$-direction of frame $S$ and the origins of
the two frames coincide at $t=0=t'$. \\
(i)  If an event occurs at $x'=60$m, $t'=8\times 10^{-8}$sec in $S'$, 
what are the spacetime coordinates of this event in $S$?\\
(ii) If an event occurs in $S$ at $t=2\times 10^{-7}$sec and $x=50$m, what time does it occur in $S'$?\\
(iii) If a second event (after (ii)) occurs at $(3\times 10^{-7}$sec,$10$m) in $S$, what is the time interval between
the events in $S'$?
\end{problem}

\begin{problem}
Two events have coordinates in the frame $S$ as follows:\\
Event 1: $x_1=x_0$, $t_1=x_0/c$ ($y_1=0=z_1$)\\
Event 2: $x_2=2x_0$, $t_2=x_0/2c$ ($y_2=0=z_2$).\\
(i) There exists a frame in which these two events are simultaneous.  Find the velocity $v$ of this
frame with respect to $S$.\\
(ii) What is the value of $t'$ for when both events occur in this $S'$?  
\end{problem}

\begin{problem} An observer does not have a complete view of what is happening
everywhere in his reference frame at a given instant; he is aware only of what is happening
at his location at that instant.  Suppose a meter stick pointing in the $x$-direction 
moves along the $x$-axis with speed $0.8c$ in frame $S$, with its midpoint passing through the
origin $0$ at time $t=0$.  Assume that an observer in frame $S$ is located at $x=0$ and 
$y=1$m. Assume $c=300,000$ km/sec.\\
(i) In $S$, where are the ends of the meter stick at time $t=0$?  ($\pm 0.3$m)\\
(ii) When does the observer see the midpoint pass the origin? ($0.33\times 10^{-8}$ sec)\\
(iii) Where do the endpoints appear to be at this time? (0.27m, -0.34m)
\end{problem}

\subsection{Time dilation and length contraction}
For a coordinate frame $S'$ moving with velocity $v$ along the $x$-axis of a rest frame $S$, we can rewrite the Lorentz transformation 
and its inverse in terms of coordinate differences to obtain
\begin{eqnarray*}\label{E2.1*}
\Delta t'&=& \gamma(\Delta t-v\Delta x/c^2)\\
\Delta x' &=& \gamma(\Delta x- v\Delta t).
\end{eqnarray*}
and
\begin{eqnarray*}\label{E2.2*}
\Delta t&=& \gamma(\Delta t'+v\Delta x'/c^2)\\
\Delta x &=& \gamma(\Delta x'+v\Delta t').
\end{eqnarray*}
Suppose we have a clock at rest in the system $S$.  Two consecutive ticks of this clock are then characterized by $\Delta x = 0$. If we want to know the relation 
between the times between these ticks as measured in both systems, we can use the first equation and find: $\Delta t'=\gamma \Delta t$ (for events in which
$\Delta x=0)$. Since $\gamma=\gamma_v=1/\sqrt{1-v^2/c^2}>1$, this shows that the time $\Delta t'$ between the two ticks as seen in 
the `moving' frame $S'$ is larger than the time $\Delta t$ 
between these ticks as measured in the rest frame of the clock. This phenomenon is called \emph{time dilation}.

Similarly, suppose we have a measuring rod at rest in the unprimed system $S$. In this system, the length of this rod is written as $\Delta x$. If we want to find the 
length of this rod as measured in the `moving' system $S'$, we must make sure to measure the distances $x'$ to the end points of the rod simultaneously in the 
primed frame $S'$. In other words, the measurement is characterized by $\Delta t' = 0$, which we can combine with the fourth equation to find the relation between 
the lengths $\Delta x$ and $\Delta x'$: 
$$\Delta x'=(1/\gamma)\Delta x ~~\mbox{for events satisfying }\Delta t'=0.$$ 
This shows that the length $\Delta x'$  of the rod as measured in the 'moving' frame $S'$ is shorter than the length $\Delta x$ in its own rest frame. This phenomenon is 
called \emph{length contraction} or \emph{Lorentz contraction}. These effects are not merely appearances; they are explicitly related to our way of measuring time intervals 
between events which occur at the same place in a given coordinate system (called ``co-local" events). These time intervals will be different in another coordinate system 
moving with respect to the first, unless the events are also simultaneous. Similarly, these effects also relate to our measured distances between separated but simultaneous 
events in a given coordinate system of choice. If these events are not co-local, but are separated by distance (space), they will not occur at the same spacial distance from 
each other when seen from another moving coordinate system.

\begin{problem} The nearest star Centauri is 4.2 light years distance from earth.  How long would it take a space ship traveling at $(2/3)c$ to reach the star, according
to its internal clock.
\end{problem}

\begin{problem}
A rocketship of proper length $l_0$ travels at  constant velocity $v$ in the positive $x$-direction relative to
frame $S$.  The nose of the ship $A'$ passes over the point $A$ on the $x$-axis at $t=0=t'$, and at that instant 
a light signal is sent in the negative $x$-direction to the tail of the ship $B'$.  \\
(i) When by rocketship time $t'$ does the signal reach the tail $B'$ of the ship? ($l_0/c$)\\
(ii) At what time $t_1$, measured in $S$, does the signal reach $B'$?\\
(iii) At what time $t_2$ in $S$ does the tail of the ship $B'$ pass $A$?
\end{problem}

\begin{problem}
In a reference frame $S$ a flash of light if emitted at position $x_1$ on the $x$-axis and 
is absorbed at $x_2=x_1+l$.  In a frame $S'$ moving with velocity $v=\beta c$ along the
$x$-axis:\\
(i) What is the spatial separation $l'$ between the point of emission and point of 
absorption of the light?\\
(ii) How much time elapses (in $S'$) between the emission and absorption of the light?
\end{problem}

\subsection{Lorentz boosts}
The preceding considerations generalize in a straightforward way to the setting 
of two frames $S$ and $S'$, both with the same space-time origin.  We assume 
that $S'$ is moving within frame $S$ at a velocity $\v$ in the coordinate system
of the frame $S$, and hence that $S$ is moving within $S'$ at velocity $-\v$ in 
the coordinate system of the frame $S'$.  Let $\x\in\R^3$ and let $x_\parallel$ 
be the orthogonal projection of $x$ onto the line $\R\v$ and $\x_\perp$ be
the orthogonal projection onto the hyperplane subspace $\v^\perp$ 
perpendicular to $\v$.  Then by a straightforward generalization of the
preceding calculations we have for $\gamma_\v:=\frac{1}{(1-\Vert\v\Vert^2/c^2)^{1/2}}$ for $\Vert \v\Vert<c$
\begin{eqnarray}
ct'&=& \gamma_\v(ct-\frac{\v\cdot\x}{c})\nonumber\\ 
\x'&=&\x_\perp+\gamma_\v(\x_\parallel-t\v)
\end{eqnarray}
By standard vector geometry $\x_\parallel=(\v\cdot \x)\v/\vert \v\vert^2$ and 
$\x_\perp=\x-(\v\cdot \x)\v/\vert \v\vert^2$; hence we can alternatively write
\begin{eqnarray}
ct'&=& \gamma_\v(ct-\frac{\v\cdot\x}{c})\nonumber\\ 
\x'&=&\x-\frac{\v\cdot \x}{\vert \v\vert^2}\v+\gamma_\v\bigl(\frac{\v\cdot \x}{\vert \v \vert^2}\v-t\v\bigr)\nonumber\\
&=&  \x-\gamma_\v t\v+\frac{\gamma_\v-1}{\vert\v\vert^2}(\v \cdot \x)\v.
\end{eqnarray}
Reversing the roles of $S$ and $S'$, we obtain
\begin{eqnarray}
ct&=& \gamma_\v(ct'+\frac{\v\cdot\x'}{c})\nonumber\\ 
\x &=&  \x'+\gamma_\v t'\v+\frac{\gamma_\v-1}{\vert\v\vert^2}(\v \cdot \x')\v.
\end{eqnarray}
The latter coordinate transformation is called the \emph{Lorentz boost along} $\v$, and can be
written in matrix notation (and unprimed coordinates) as
\begin{eqnarray}\label{E:boost}
B(\v)\left[\begin{matrix}ct\\ \x\end{matrix}\right]=\left[\begin{matrix}\gamma_\v(ct+\frac{\v\cdot \x}{c}) \\
\x+\gamma_\v t\v+\frac{\gamma_\v-1}{\vert \v\vert^2} (\v\cdot \x)\v\end{matrix}\right]
\end{eqnarray}

\begin{problem} Suppose that the frame $S'$ is moving within frame $S$ at velocity
 $\v=\langle 1,1,1\rangle$.  Write out the equations for $ct$, $x$, $y$, and $z$ in terms
 of $ct'$, $x'$, $y'$, $z'$.
\end{problem}

\begin{problem}
Find a matrix representation for the Lorentz boost $B(\v)$.
\end{problem}

\subsection{Minkowski spacetime}
In physics and mathematics, Minkowski space (or Min\-kowski spacetime) is the standard mathematical setting for Einstein's theory of special 
relativity. In this setting the three ordinary dimensions of space are combined with a single dimension of time to form a four-dimensional 
manifold for representing a spacetime. Minkowski space is named after the German mathematician Hermann Minkowski, who introduced
it in 1908.

\emph{Minkowski spacetime} $\M$  is defined to be the vector space $\R^4$ equipped with the symmetric bilinear form
\begin{equation}
\eta\bigg(\left[\begin{matrix} ct\\ \x\end{matrix}\right],\left[\begin{matrix} ct'\\ \x'\end{matrix}\right]\bigg)=-c^2tt'+\x\cdot \x', \mbox{ for } t,t'\in \R,
\ \x,\x'\in\R^3.
\end{equation}
If we define
$$I_{1,3}=\left[\begin{matrix}-1 &0& 0 & 0\\
0 &1 & 0 & 0\\
0 & 0 & 1 & 0\\ 0 & 0 & 0& 1
\end{matrix}\right],$$
then we may alternatively define $\eta(\mathbf{a},\mathbf{b})=\langle \mathbf{a},I_{1,3}\mathbf{b}\rangle$, where $\langle \mathbf{a}, \mathbf{b}\rangle= \mathbf{a}\cdot  \mathbf{b}$, the usual
euclidean inner product.  We call $\eta$ the \emph{Lorentzian form} of the spacetime $\M$.

A linear transformation $T:\R^4\to\R^4$ is said to \emph{preserve the form} $\eta$ if $\eta(T\x,T\y)=\eta(\x,\y)$ for all $\x,\y\in\R^4$.   it follows from
\begin{equation}\label{E:Lequiv}
\eta(T\x,T\y)=\langle Tx,I_{1,3}T\y\rangle=\langle \x,T^TI_{1,3}Ty\rangle
\end{equation}
that $T$ preserves $\eta$ iff $I_{1,3}=T^TI_{1,3}T$, where $T^T$ is the adjoint of $T$.  We summarize:
\begin{proposition}\label{P:eta}
 A linear transformation $T:\R^4\to\R^4$ preserves $\eta$ iff $I_{1,3}=T^TI_{1,3}T$ iff $I_{1,3}T^{-1}=T^TI_{1,3}$.
\end{proposition}

\begin{problem}\label{P:2.12}
 Show that the Lorentz boost  
 \begin{equation}\label{E:boost2}
B(\v)=
\left[\begin{matrix}\gamma_\v &\frac{\gamma _\v \v^T}{c}\\
\frac{\gamma_\v \v}{c} & I+\frac{\gamma_\v-1}{\vert \v\vert^2}\v\v^T
\end{matrix}\right]
\end{equation}
preserves $\eta$.  (Hint: Use the preceding proposition and
the fact that $(B(\v))^{-1}=B(-\v)$.)
\end{problem}

Thus while Lorentz boosts do not preserve time or distance, they do preserve the form $\eta$, i.e., $\eta$ is an
invariant for Lorentz boosts.  This fact can frequently simplify special relativity calculations.

\begin{problem}
Two events occur at the same place and $4$ seconds apart in inertial frame $S$.  What is their spatial separation in a frame $S'$ in which the events are $6$ seconds
apart?
\end{problem}

\begin{problem} Two events occur at the same time in an inertial frame $S$ and are separated by a distance of $1$ km along the $x$-axis.  What is the time difference
between these two events as measureed in a frame $S'$ moving with constant velocity along the $x$-axis for which the spatial separation of the two events 
is measured as $2$ km.
\end{problem}

\subsection{Minkowski diagrams revisited}
We have seen earlier that in the Minkowski diagram with coordinate axes determined by the frame $S$, a frame $S'$ moving along the $x$-axis at velocity $v$
has one axis $x=vt=(v/c)(ct)$ (corresponding to $x'=0$) and the other axis the line through the origin parallel to any line of simultaneity (corresponding to 
$t'=0$).

\begin{problem}
Find an equation for the second axis.  Show that the two axes for $S'$ make the same angle with the diagonal, one on each side of it.
\end{problem}

In $S'$ the unit along the $x'$-axis has coordinates $x'=1$, $t'=0$.  Thus $-(ct')^2+(x')^2=1$.  Since the Lorentz transformation to $S$ must preserve the
Lorentzian form, it follows that the unit on the $x'$-axis must be the intersection of that axis with
the hyperbola $-(ct)^2+x^2=1$ or $x^2-(ct)^2=1$.  A similar calculation yields the unit length on the $t'$-axis.  Thus to read off the space and time 
$t',x'$-coordinates of a given point event $P$, draw lines through $P$ parallel to the $x'$- and $t'$-axes, and read off the intercepts.

\begin{problem}
Two reference frames $S$ and $S'$ move with speed $c/2$ with respect to each other.  
\item{(a)} Draw a Minkowski diagram relating the two systems (let
the $x$ and $ct$ axes be perpendicular).  Draw the calibration hyperbolas that allow you to define distance on the $ct'$ and $x'$ axes.
\item{(b)} Plot the following points on the diagram:  (1) $x=1$, $ct=1$, (2) $x'=1$, $ct'=1$, (3) $x'=2$, $ct'=0$, (4) $x=0$, $ct=2$.
\item{(c)} From your diagram determine the coordinates in the other coordinate system for each point plotted in (b).
\end{problem}

\section{The Lorentz Group}
In the preceding section we have considered transformations between certain reference frames moving at a constant
velocity with respect to one another.  These transformations consist of what we called Lorentz boosts, and we saw that they
preserved the Lorentzian form $\eta$.   The invertible linear transformations on $\R^4$ that preserve $\eta$ form
a group under composition, usually referred to as the generalized orthogonal group $O(1,3)$.  This is also frequently called
the Lorentz group, but we prefer to define the \emph{Lorentz group} to be the subgroup of $O(1,3)$ of \emph{time-preserving}
or \emph{orthochronous} $\eta$-preserving invertible linear transformations.  One characterization of time preservation
is that the vector $(1,0,0,0)\in\R^4$ is carried into a vector with a positive $t$-component.  Such transformations we shall call 
\emph{Lorentz transformations} or \emph{linear isometries of Minkowski space}.  We denote the Lorentzian group
by $O^+(1,3)$.  We extend our notion of inertial or reference frame to include frames arising from the usual coordinates of
$\R^4$ by applying a Lorentz transformation.  

\begin{problem}\label{P:3.1}  Show that a Lorentz boost is a Lorentz transformation in the preceding sense.  Show that 
a transformation of $\R^4$ that leaves the time coordinate fixed and acts as an orthogonal transformation
on the space coordinates is a Lorentz transformation.
\end{problem}

In physics and mathematics, the Lorentz group is the group of all Lorentz transformations of Minkowski spacetime, 
the special relativistic setting for all (nongravitational) physical phenomena. The mathematical form of
standard physical laws such as the  kinematical laws of special relativity, Maxwell's field equations in the 
theory of electromagnetism, and Dirac's equation in the theory of the electron,
are each invariant under Lorentz transformations. Therefore the Lorentz group can be said to express a 
fundamental symmetry of many of the known fundamental laws of nature.

\subsection{Spacetime intervals and causality}
The \emph{spacetime interval} $s^2$ between two events is defined to be $s^2=-(c\Delta t)^2+(\Delta \x)^2$, where $\Delta \x$ is the distance between the space
coordinates.  Note that the spacetime interval is invariant under any Lorentz transformation.    If $s^2<0$, then the interval is said to
be \emph{time-like} and the \emph{proper time} of the interval is defined to be $\sqrt{(\Delta t)^2-(\frac{\Delta \x}{c})^2}$.  If $s^2=0$, the interval is said to be \emph{light-like},
and if $s^2>0$, the interval is said to be \emph{space-like} and the \emph{proper distance} is defined to be $\sqrt{(\Delta \x)^2-(c\Delta t)^2}$.  Similarly for
an individual element $\u\in\M$, the element is called time-like, resp.\ light-like, resp.\ space-like depending on whether $\eta(\u,\u)$ is less than $0$, resp.\ equal to
$0$, resp.\ greater than $0$.  

We note that the light-like elements form a double cone, each cone having a circular cross section.  The time-like elements consist of two connected components, one making up
the interior of the top light cone and the other the bottom.  The time-like and light-like elements together with $t$-coordinate greater than or equal to $0$ make up a closed convex
cone $K$ with dense interior Int$(K)$ made up by the time-like vectors.  Such cones are called Lorentzian cones.  

The cone $K$ induces an order on spacetime $\M$ defined by $\u\leq \v$ if $\v-\u\in K$.  This order is a partial order called the \emph{causal order}.  We write
$\u<\v$ if $\u\ne \v$ and $\u\leq \v$ and $\u\prec\v$ if $\v-\u\in\mbox{Int}(K)$.  If $\u\leq \v$, then we say that $\u$ has a \emph{potential causal connection}
with $\v$.

\begin{problem}  Show that the causal order is a partial order. \end{problem}

\begin{problem}\label{P:3.3} Show that the sets of time-like, space-like, and light-like vectors respectively are preserved by members of $O(1,3)$. \end{problem}

\begin{problem}\label{P:3.4}  Argue that the time-like vectors consist of two connected components.  Argue that an $\eta$-preserving linear transformation is orthochronous iff
it preserves the interior of $K$. 
\end{problem}

\begin{problem}  Show that $K$ and Int$(K)$ are preserved by Lorentz transformations.  
Show that the causal order and the order $\prec$ are preserved by Lorentz transformations.  
\end{problem}

\begin{problem} Argue that if $\u<\v$, then in any reference frame the $t$-coordinate of $\u$ is less than the $t$-coordinate of $\v$.
\end{problem}

\begin{problem} An object is moving with a constant (admissible) velocity $\v$ in a reference frame $S$.  Argue that events on its world line are getting larger
in the causal order as $t$ grows.
\end{problem}

\begin{problem}
Answer both questions for the two events in each part:  Is there a  potential causal connection between the events?  Is there a frame in which the two events
are simultaneous?\\
(a) $(2\times 10^{-9}sec,0.3m,0.5m,0)$ and $(3\times 10^{-9}sec,0.4m,0.7m,0)$.\\
(b) $(5\times 10^{-9}sec,0.7m,0.5m,0)$ and $(4\times 10^{-9}sec,0.4m,0.6m,0)$.
\end{problem}

\subsection{Symmetric matrices}
\emph{Throughout this section and the next all matrices are assumed to be square
matrices with real entries.}
A matrix $A$ is \emph{symmetric} if $A=A^T$, where $A^T$ 
denotes the transpose.  
Let ${\mathrm{Sym}}(n,{\R})$, or simply $\Sym$ when $n$ is understood,
be the vector space of all $n\times n$ symmetric 
matrices. For $A\in\Sym,$ we recall that  
$A$ is positive semidefinite, denoted $0\leq A$, if $x^TAx=\langle x,Ax\rangle\ge 0$ for
all $x\in\R^n$, where $\la\cdot,\cdot\ra$ denotes the usual inner product
on $\R^n$.  Similarly $A$ is positive definite, denoted $0<A$, if it is 
positive semidefinite and invertible, or equivalently if $x^TAx=\la x,Ax\ra>0$ for
all non-zero $x$. We denote the set of positive definite (semidefinite)
matrices by $\Spp$ ($\Sp$).

The following ``internal" characterization of a positive definite matrix
involves \emph{orthogonal matrices}, matrices $U$ such that $U^{-1}=U^T$, and 
is a standard linear algebra result.

\begin{proposition}\label{P:beg}
A symmetric matrix $A$ is positive definite (semidefinite) if and only if 
$A=U^TDU$ for some orthogonal matrix $U$ and some diagonal matrix 
$D$ with positive (non-negative) diagonal
entries if and only if $A$ has all eigenvalues positive (non-negative).
\end{proposition}

By standard matrix spectral theory, we can write  $A\in\Sym(n,\R)$ uniquely as 
$A=\sum_{i=1}^r\lambda_iE_i$,  where $\lambda_1,\ldots,\lambda_r\in\R$ 
are distinct, each $E_i$ is a non-zero orthogonal projection, 
and the collection
$\{E_i\}_{1\le i\le  r}$ satisfies $\sum_{i=1}^r E_i=I$ and $E_iE_j=0$ 
for $i\ne j$.  Indeed the existence follows by choosing the $\lambda_i$ to
be the distinct eigenvalues and each $E_i$ to be the orthogonal projection
onto the eigenspace of $\lambda_i$.  The uniqueness, on the other hand,
follows by (i) observing that for any such decomposition 
the $\lambda_i$, $i=1,\ldots,r$, must be eigenvalues 
and the range of $E_i$ must consist of eigenvectors for $\lambda_i$, 
(ii) using the equality
$\sum_{i=1}^r E_i=I$ and the orthogonality $E_iE_j=0$ for $i\ne j$
to argue that $\R^n$ is the direct sum of the ranges of the $E_i$, 
and then (iii) deducing that the ranges of the $E_i$ must exhaust the
eigenspaces and thus that the $\lambda_i$ must exhaust the eigenvalues. We call
$\{\lambda_1,\ldots,\lambda_r\}$ the \emph{spectrum} of $A$ and
$\sum_{i=1}^r\lambda_iE_i$ the \emph{spectral decomposition}.

\begin{problem}
Use the spectral decomposition to construct a proof of the representation
of a positive definite matrix given in Proposition \ref{P:beg}.
(Hint: Pick an orthonormal basis of eigenvectors and consider a change of
coordinates between it and the standard basis.)
\end{problem}

For a bijection $f:M_1\to M_2$, where $M_1,M_2\subseteq \R$, 
we define a function on all $A\in\Sym$ with
spectrum contained in $M_1$ by $f(A)=\sum_{i=1}^r f(\lambda_i)E_i$, where
$A=\sum_{i=1}^r\lambda_iE_i$ is the spectral decomposition (functions 
constructed in this way are called \textit{matrix functions} and  
provide a simple example of the functional calculus). 
Note (from uniqueness of spectral
decomposition) that $f$ is well-defined and defines a bijection
from all symmetric matrices with spectrum contained in $M_1$ to all 
symmetric matrices with spectrum contained in $M_2$ 
(with inverse defined from $f^{-1}:M_2\to M_1$).  

Extending $\exp:\R\to (0,\infty)$ to $\Sym$, 
we obtain (in light of Proposition \ref{P:beg}) that

\begin{proposition}\label{P:inj}
The exponential function $\exp:\Sym\to\Spp$ given by
$$\exp A=\exp(\sum_{i=1}^r\lambda_iE_i)=\sum_{i=1}^r e^{\lambda_i}E_i$$
is a bijection. 
In particular, a symmetric matrix is positive definite
if and only if it is the exponential of a symmetric matrix.
\end{proposition}

We recall that the matrix exponential function is more
commonly defined by the power series $e^A=\sum_{n=0}^\infty A^n/n!$,
which converges for all $A$.  Since for $x$ in the eigenspace
of an eigenvalue $\lambda_i$ of $A$, we have
$e^A(x)=e^{\lambda_i}x=\exp A(x)$, the matrix operators 
$e^A$ and $\exp A$ agree on the eigenspaces of $A$, and
hence $e^A=\exp A$.

Using the same methods as those employed for Proposition \ref{P:inj},
we obtain by extending the bijection
$f(x)=x^2$ on $(0,\infty)$ (respectively, $[0,\infty)$\,) to the matrices with
spectrum contained in $(0,\infty)$ (respectively, $[0,\infty)$\,),
that is, the positive 
definite (respectively, semidefinite)
matrices, the following

\begin{proposition}\label{P:sqrt}
If $A>0$ (respectively, $A\ge 0$), then $A$ has a unique positive definite 
(respectively, semidefinite) square root, denoted $A^{1/2}$.
\end{proposition}

For any $A>0$, by Proposition \ref{P:inj}
there exists a unique symmetric matrix
$\log A$ such that $\exp(\log A)=A$.
We can thus define $A^r$ for
any $A>0$ by $A^r=\exp(r\log A)$.  
Since $\exp(A+B)=\exp(A)\exp(B)$ if $AB=BA$, the one-parameter group 
$\{A^r:r\in\R\}$ satisfies the standard laws of exponents.  This leads to the
\begin{proposition}\label{P:exp}
The map $\exp$ satisfies 
$$A^{t+s}=\exp((t+s)X)=\exp(tX+sX)=\exp(tX)\exp(sX)=A^tA^s$$
for any $X\in\Sym$ and $A\in\Spp$ satisfying $A=\exp X$
$($equivalently, $X=\log A)$.
\end{proposition}  
For $A$ positive definite and $r=1/2$, the preceding definition for $A^{1/2}$
agrees with that of Proposition \ref{P:sqrt}
since $\exp ((1/2)\log A)$ is a positive definite square root of $A$ and this
square root is unique.

\begin{problem}
Use the fact that $\exp:\Sym\to\Spp$ is bijective and a homomorphism on
each one-dimensional subspace to show that each member of $\Spp$ has a 
unique $n$th root in $\Spp$.
\end{problem}

%Let $\leq$ be  the partial order (sometimes called the Loewner order) on 
%$\Sym$ induced by the closed convex cone of positive 
%semidefinite matrices, i.e., $A\le B$ if and only if $B-A$ is positive
%semidefinite. Also we define $A<B$ if $B-A>0$.

\subsection{Polar decompositions}
A \emph{polar decomposition} for an invertible matrix $A$ is a factorization
$A=PU$, where $P$ is a positive definite matrix and $U$ is an orthogonal
matrix.  
\begin{proposition}\label{P:polar} (Polar Decomposition)  Each invertible
matrix $A$ has a unique polar decomposition $A=PU$.  Furthermore,
$P=\sqrt{AA^T}$ is the unique positive definite square root of $AA^T$.
\end{proposition}

\begin{proof}.  Since $A$ is invertible, $\la \x,AA^T \x\ra=\la A^T\x,A^T\x\ra
>0$ for all $\x\ne 0$, and hence $AA^T=(AA^T)^T$ is positive definite.  
By Proposition \ref{P:sqrt} $AA^T$ has a unique positive definite square
root $P$.   Set $U:=P^{-1}A$.  Clearly $PU=A$.  Furthermore,
$$UU^T=P^{-1}AA^T(P^{-1})^T=P^{-1}P^2P^{-1}=I,$$
so $U$ is orthogonal.

Suppose $A=QV$ is another polar decomposition.
Then $AA^T=QVV^TQ=Q^2$.  Since positive definite square
roots are unique, $Q=P$, and hence $V=Q^{-1}A=P^{-1}A=U$.
\end{proof}

We want to show next that if $A$ is a Lorentz transformation, then so are the 
polar factors.  First, a lemma.
\begin{lemma}\label{L:Lroots}
Let $P$ be positive definite matrix preserving the Lorentzian form.  Then $P$
and all of its powers $P^t$, $t\in \R$, are Lorentz transformations.
\end{lemma}

\begin{proof} By Proposition \ref{P:beg} there exists an orthogonal $U$ such that
$P=U^{T}DU$, where $D$ is a diagonal matrix with positive entries down the
diagonal.  By Proposition \ref{P:eta} $PI_{1,3}P=I_{1,3}$.  Thus
$$ D(UI_{1,3}U^T)D=UPU^TUI_{1,3}U^TUPU^T=UPI_{1,3}PU^T=UI_{1,3}U^{T}.$$
If we set $B=UI_{1,3}U^T$, then $DBD=B$.  The only way this can happen for
a diagonal matrix $D$ is that $d_{i,i}d_{j,j}=1$ whenever $b_{i,j}\ne 0$.
It then follows for $t\in\R$ that $d_{i,i}^td_{j,j}^t=1$ whenever $b_{i,j}\ne 0$, and thus
$D^{t}BD^{t}=B$.  Reversing our earlier argument, we conclude that
$U^TD^{t}UI_{1,3}U^TD^{t}U=I_{1,3}$.  It is straightforward that 
$U^TD^{t}U$ provides an alternative way for computing $P^t$, and
thus $P^t$ preserves the Lorentzian form.

Let $\mathbf t$ denote the unit vector in the $ct$-direction in Minkowski space.
Since all $P^t$ preserve the Lorentzian form, they must carry the vector $\mathbf t$  into the open positive cone of
time-like vectors or its negative.  Since the map from $\R$ to $\R$ given by $t\mapsto \pi_{ct}(P^t(\mathbf{t}))$ is continuous,
where $\pi_{ct}$ is projection into the $ct$-coordinate, takes the value $1$ at $t=0$, and can't assume the
value $0$, by the Intermediate Value Theorem, we conclude that it takes on only positive values.  Thus all
$P^t$ are Lorentz transformations.
\end{proof}

\begin{problem} Show in the previous proof that $U^TD^{t}U$ gives the power $P^t$, as defined in Section  3.2.
\end{problem}

\begin{lemma}\label{L:trans}
If $A\in O(1,3)$, so is $A^T$.
\end{lemma}

\begin{proof} If $A^TI_{1,3}A=I_{1,3}$, then taking inverses we obtain
$A^{-1}I_{1,3}(A^T)^{-1}=I_{1,3}$.  Since $A^{-1}I_{1,3}(A^T)^{-1}=
((A^{-1})^T)^TI_{1,3}(A^{-1})^{T}$, we conclude that $(A^{-1})^T$
preserves the Lorentzian form.  Now if $A\in O(1,3)$,
so also is $A^{-1}$, and applying the preceding to $A^{-1}$, we conclude
that $A^T\in O(1,3)$. 
\end{proof}

\begin{proposition}\label{P:Ldecomp}
The polar factors $P,U$ of the polar decomposition $A=PU$ of a Lorentz transformation
$A$ are also Lorentz transformations.
\end{proposition}
\begin{proof}
By the preceding lemma, $A^T\in O(1,3)$, so $AA^T\in O(1,3)$.  Since $AA^T$ is positive
definite, it follows from Lemma \ref{L:Lroots} that it  and $P=(AA^T)^{1/2}$ (from Proposition
\ref{P:polar}) are Lorentz transformations.  Thus $U=P^{-1}A$ is a Lorentz transformation.
\end{proof}

We close this section by characterizing those orthogonal matrices that are Lorentz
transformations.
\begin{proposition}\label{P:orthfact}  An orthogonal matrix $U$ is a Lorentz transformation  iff it has the block form
\[  
U= \left[\begin{matrix} 1 & 0\\
0 & S
\end{matrix}\right],
\]
where $S\in O(3)$, i.e., is an orthogonal transformation on $\R^3$.
\end{proposition}
\begin{proof}  
The proof follows directly from the equation $U^TI_{1,3}U=I_{1,3}$, or equivalently,
$I_{1,3}U=UI_{1,3}$.  The $1,2$-blocks on each side of the latter equation are negatives
of each other and hence must be $0$, and ditto for the $2,1$-blocks. Hence
$U$ is block diagonal with each block having its transpose for its inverse.
Since $U$ is a Lorentz transformation, it follows that  the $1,1$-entry must be positive and its
own inverse, hence $1$.  The block $S$ must also have its transpose being its inverse,
hence must be orthogonal.
\end{proof}  

\subsection{Positive definite Lorentz transformations}
The goal of this section to to show that the positive definite Lorentz transformations are 
precisely the Lorentz boosts.  It follows from Problems \ref{P:2.12} and \ref{P:3.1}
that Lorentz boosts are symmetric Lorentz transformations, so we need to establish
the converse.  

We write an arbitrary positive definite Lorentz transformation in the form
\[  
A= \left[\begin{matrix} \tau &\ \x^T\\
\x & S
\end{matrix}\right],
\]
where $\tau$ is a positive scalar (since $A$ is a Lorentz transformation), $S$ is a $3\times 3$ symmetric
matrix (since $A$ is symmetric), and $\x\in\R^3$ is a column vector (by the symmetric property its transpose 
row vector $\x^T$ must appear after $\tau$ as the remainder of the first row).  Since $A$ is positive definite,
we have for $0\ne\y\in\R^3$,
$$0<\left[\begin{matrix} 0 & \y^T\end{matrix}\right]  \left[\begin{matrix} \tau &\ \x^T\\
\x & S\end{matrix}\right] \left[\begin{matrix} 0 \\ \y\end{matrix}\right]=\y^TS \y,$$
from which we conclude that $S$ is positive definite.

We recall from 
Proposition \ref{P:eta}  (and the fact that $A$ is symmetric) that
\[
\left[\begin{matrix} \tau &\ \x^T\\ \x & S \end{matrix}\right]
\left[\begin{matrix} -1& \mathbf{0}\\ \mathbf{0} & I\end{matrix}\right]
\left[\begin{matrix} \tau &\ \x^T\\ \x & S \end{matrix}\right]
= \left[\begin{matrix} -1& \mathbf{0}\\ \mathbf{0} & I\end{matrix}\right]
\]
If we multiply out the left-hand side and set the $1,1$-entries equal on both
sides of the equation, we obtain $-\tau^2+\x^T\x= -1$, which implies
$\tau=\sqrt{1+\x^T\x}$.  We set
$$\v:=\frac{c}{\sqrt{1+\x^T\x}}\x.$$
Then
$$\frac{\vert \v\vert^2}{c^2}=\frac{\x^T\x}{1+\x^T\x}<1;\mbox{ hence } \vert \v\vert^2<c^2,\mbox{ i.e., } \vert \v\vert <c.$$
We next compute
\begin{equation}
\gamma=\gamma_\v=\frac{1}{\sqrt{1-\frac{\v^T\v}{c^2}}}=\frac{1}{\sqrt{1-\frac{\x^T\x}{1+\x^T\x}}}=\sqrt{1+\x^T\x}=\tau.
\end{equation}
Thus we obtain
\[
A=\left[\begin{matrix} \gamma_\v&\frac{\gamma _\v \v^T}{c}\\
\frac{\gamma_\v \v}{c} &S
\end{matrix}\right]
\]

By equating the $2,2$-entries in the equation $AI_{1,3}A=I_{1,3}$, we obtain
the following string of equivalent equalities:
\begin{eqnarray}\label{E:Sdef}
-\frac{\gamma^2}{c^2}\v\v^T+S^2 &=& I\nonumber\\
S^2 &=& I+\frac{\gamma^2}{c^2} \v\v^T\nonumber\\
S &=& \sqrt{I + \frac{\gamma^2}{c^2}\v\v^T}.
\end{eqnarray}
The last equation follows from the fact that $S$ is positive definite, so $S^2$ is, and thus has unique positive square root $S$
(Proposition \ref{P:sqrt}).

We record the following useful identity, which can be directly verified from the definition of $\gamma_\v$, and which we use in the last
step of the following calculation.
\begin{equation}\label{E:useful}
\gamma_{\v}^2-1=\frac{\vert \v\vert^2}{c^2}\gamma_{\v}^2.
\end{equation}
We now calculate
\begin{eqnarray*}
\bigg(I+\frac{\gamma_\v -1}{\vert \v\vert^2}\v\v^T\biggl)^2 &=& I+2\frac{\gamma_\v-1}{\vert\v\vert^2}\v\v^T+\frac{(\gamma_\v-1)^2}{\vert \v\vert^4}\v\v^T\v\v^T\\
&=& I+\frac{2\gamma_\v-2}{\vert\v\vert^2}\v\v^T+\frac{\gamma_{\v}^2-2\gamma_\v+1}{\vert \v\vert^2}\v\v^T\\
&=& I+\frac{\gamma_\v^2-1}{\vert\v\vert^2}\v\v^T\\
&=& I+\frac{\gamma_\v^2}{c^2}\v\v^T
\end{eqnarray*}
It follows from this calculation and Equation (\ref{E:Sdef}) that 
\[
A= \left[\begin{matrix} \gamma_\v&\frac{\gamma _\v \v^T}{c}\\
\frac{\gamma_\v \v}{c} &S
\end{matrix}\right]=\left[\begin{matrix} \gamma_\v&\frac{\gamma _\v \v^T}{c}\\
\frac{\gamma_\v \v}{c} & I+\frac{\gamma_\v-1}{\vert \v\vert^2}\v\v^T
\end{matrix}\right]
\]
We have thus established the
\begin{proposition}\label{P:eqboost}
A positive definite Lorentz transformation is a Lorentz boost.  
\end{proposition}

\section{Einstein velocity addition}
The notion of velocity addition arises in at least two obvious contexts in special relativity.  The first is the problem of finding the velocity of an object
in a frame $S'$, given that it is moving at some constant velocity in a reference frame $S$.  The second is the problem of finding the velocity at
which $S''$ is moving with respect to $S$, given the knowledge of how fast $S'$ is moving with respect to $S$ and $S''$ is moving with respect to $S'$.
The problems are more-or-less interchangeable since one can pass from an object to a frame in which it is stationary, and from a frame to an object
stationary in that frame.  There is the caveat, however, that for an object moving at constant velocity in one frame, there are more than one such frames
in which it is stationary  (for any such frame, consider the new frame obtained by a rotation of its space coordinates).  

We suppose first that the frame $S'$ is moving with velocity $v$ in the $x$-direction in the reference frame $S$.  We set 
$\gamma=\gamma_v=(1-v^2/c^2)^{-1/2}$.  We then have the Lorentz boost given by the equations
$$t=\gamma(t'+vx'/c^2),\qquad x=\gamma(x'+vt'),\qquad y=y',\qquad z=z'.$$
Suppose that an object has velocity components $u_x',u_y',u_z'$ as measured in frame $S'$.  This means by definition that
$$u_x'=\frac{dx'}{dt'},\qquad u_y'=\frac{dy'}{dt'},\qquad u_z'=\frac{dz'}{dt'}.$$
Differentiating the equations of the Lorentz boost, we obtain
\begin{eqnarray}\label{E:Lvel}
u_x &=& \frac{dx}{dt}=\frac{dx/dt'}{dt/dt'}=\frac{\gamma(dx'/dt'+v)}{\gamma(1+v(dx'/dt')/c^2)}=\frac{u_x'+v}{1+vu_x'/c^2}\nonumber\\
u_y&=&\frac{dy}{dt}=\frac{dy/dt'}{dt/dt'}=\frac{u_y'}{\gamma(1+vu_x'/c^2}=\frac{u_y'/\gamma}{1+vu_x'/c^2}\\
u_z&=&\frac{dz}{dt}=\frac{dz/dt'}{dt/dt'}=\frac{u_z'}{\gamma(1+vu_x'/c^2}=\frac{u_z'/\gamma}{1+vu_x'/c^2}.\nonumber
\end{eqnarray}

\begin{problem} Suppose that $v=u_x'=0.5c$ and $u_y'=u_z'=0$.  What is $u_x$?  What would it be in Newtonian mechanics?
\end{problem}

In a completely analogous way one obtains
\begin{eqnarray}\label{E:Lvel2}
u_x'=\frac{u_x-v}{1-vu_x/c^2}, \qquad u_y'=\frac{u_y/\gamma}{1-vu_x/c^2},\qquad u_z'=\frac{u_z/\gamma}{1-vu_x/c^2}.
\end{eqnarray}

\subsection{The one-dimensional case}
The simplest case to consider is the case that the velocity $\v$ and the velocity $\u$ are both in the direction of the $x$-axis.   
We consider a frame $S'$ moving with velocity $v$ along the $x$-axis with respect to a frame $S$ and a frame $S''$ moving with
velocity $u$ along the $x$-axis with respect to $S'$.  We assume that all three frames share a common origin.  We may assume that
$S''$ is the frame at which some object moving with velocity $u$ in $S'$ is at rest.  Then by equation (\ref{E:Lvel}) we have that
$S''$ is moving with respect to $S$ with velocity $w:=(v+u)/(1+vu/c^2)$.  We therefore define the Einstein velocity addition of
$v$ and $u$ for $u,v\in\R_c=\{u\in\R:\vert u\vert<c\}$ by
\begin{equation}
v\oplus u= \frac{v+u}{1+\frac{vu}{c^2}}.
\end{equation}
Note that the formula gives the velocity $v\oplus u$ with which a third frame or object is traveling with respect to a first frame or object,
given that a second frame or object is traveling with velocity $v$ with respect to a first frame or object, and a third
with velocity $u$ with respect to the second.
\begin{problem}  Suppose that rockets A and B are speeding toward each other at speeds $0.8c$ for A and $0.9c$ for B, both
speeds calculated in reference frame $S$.  Suppose that rocket A fires a missile toward rocket B at a velocity $0.7c$ with
respect to the frame of A.  How fast is the missile traveling in the original reference frame and in the reference frame of B?
\end{problem}

\begin{problem} A $K^\circ$ meson at rest decays into a $\pi^+$ meson and a $\pi^-$ meson, each having a speed of 0.85c.  If a $K^\circ$
 meson traveling at a speed of 0.9c in frame $S$ decays, what is the greatest  speed that one of the $\pi$ mesons can have
(again in $S$)?  What is the least speed.?
\end{problem}

\begin{problem} Show that the map $f:\R\to \R_c$ defined by $f(x)=c\tanh(x)$ is an isomorphism from $(\R,+)$ to $(\R_c,\oplus)$.
\end{problem}

\subsection{A general definition of velocity addition}
In this section we turn to the general definition of Einstein velocity addition for velocities in the open ball $\R_c^3$ of radius $c$ in $\R^3$.
Recall that the world line of an object A moving with respect to a reference frame $S'$ with constant velocity $\v$ is given by
$$\biggl\{\left[\begin{matrix} ct'\\ t'\v \end{matrix}\right]: t'\in\R\biggr\}.$$
We suppose further that the frame $S'$ is moving at constant velocity $\u$ with respect to the frame $S$.  We can calculate the equation
of the world line of A in frame $S$ via the Lorentz boost for change of coordinates from $S'$ to $S$:
\begin{eqnarray*}
\left[\begin{matrix} \gamma_\u&\frac{\gamma _\u \u^T}{c}\\
\frac{\gamma_\u \u}{c} &I+\frac{\gamma_\u-1}{\vert \u\vert^2}\u\u^T
\end{matrix}\right] \left[\begin{matrix} ct'\\ t'\v \end{matrix}\right] &=&
\left[\begin{matrix} \gamma_\u\big(ct'+\frac{t'\u^T\v}{c}\big)\\ \gamma_\u t'\u+t'\v+\frac{\gamma_\u-1}{\vert\u \vert^2}t'\u\u^T\v \end{matrix}\right]\\
&=& \left[\begin{matrix} ct'\big(\gamma_\u\big(1+\frac{\u^T\v}{c^2}\big)\big]\\ t'\big(\gamma_\u \u+\v+\frac{\gamma_\u-1}{\vert\u \vert^2}(\u^T\v)\u\big) \end{matrix}\right]
\end{eqnarray*}
We conclude that the image of the world line of A under the Lorentz boost is the world line
\begin{equation*}
\bigg\{\left[\begin{matrix} ct\\ t(\u\oplus\v)\end{matrix}\right]:t\in \R\bigg\}
\end{equation*}
where
\begin{equation*}
t=\gamma_\u(1+\frac{\u^T\v}{c^2})t'
\end{equation*}
and
\begin{equation}\label{E:einplus}
\u\oplus\v:=\frac{1}{1+\frac{\u^T\v}{c^2}}\bigg(\u+\frac{1}{\gamma_\u}\v+\frac{\gamma_\u-1}{\gamma_\u\vert\u\vert^2}(\u^T\v)\u\bigg)
\end{equation}
Using the equation $c^2(\gamma^2-1)=\gamma^2\vert\u\vert^2$ (equation (\ref{E:useful})), where $\gamma=\gamma_\u$, we note that
$$\frac{c^2(1+\gamma)}{\gamma}=\frac{c^2(1+\gamma)}{\gamma}\cdot\frac{\gamma-1}{\gamma-1}=\frac{c^2(\gamma^2-1)}{\gamma(\gamma-1)}
=\frac{\gamma^2\vert\u\vert^2}{\gamma(\gamma-1)}=\frac{\gamma\vert\u\vert^2}{\gamma-1}.$$
Inverting, we conclude 
\begin{equation}
\frac{\gamma_\u}{c^2(1+\gamma_\u)}=\frac{\gamma_\u-1}{\gamma_\u\vert\u\vert^2}.
\end{equation}
This allows us to rewrite the definition of Einstein velocity addition in the form 
\begin{equation}\label{E:einplusb}
\u\oplus\v:=\frac{1}{1+\frac{\u\cdot\v}{c^2}}\bigg(\u+\frac{1}{\gamma_\u}\v+\frac{\gamma_\u}{c^2(1+\gamma_\u)}(\u\cdot\v)\u\bigg)
\end{equation}

In the case that $\u$ and $\v$ are parallel (i.e., one is a scalar multiple of the other), Einstein addition reduces to 
\begin{equation}\label{E:einplusc}
\u\oplus\v=\frac{\u+\v}{1+\frac{\u\cdot\v}{c^2}}
\end{equation}
\begin{problem}
Prove equation (\ref{E:einplusc}). (Hint: Let $\u=r\w$, $\v=s\w$ for some $\w$, apply equation (\ref{E:einplus}) and reduce.)
\end{problem}
\begin{problem}  Alternatively prove for $-c<r,s<c$ that $$\frac{r\u}{\vert\u\vert}\oplus\frac{s\u}{\vert \u\vert}=\frac{(r+s)\u}{(1+\frac{rs}{c^2})\vert\u \vert}.$$
\end{problem}
\begin{problem} \label{oneparam}
Prove that $F:(\R,+)\to(\R_c(\x/\vert\x\vert),\oplus)$ defined by $F(x)=c\tanh(x)$ is an isomorphism.  Hence $\oplus$ restricted to any $\R_c\x$ is a group operation
isomorphic to $(\R,+)$.
\end{problem}

\begin{problem}  Show that if $\u$ and $\v$ are orthogonal, then $\u\oplus\v=\u+\gamma_\u^{-1}\v$.  In particular, the operation
 $\oplus$ is not commutative.
\end{problem}

\subsection{Einstein addition and Lorentz boosts}
This is a very close and useful connection between Einstein velocity addition, Lorentz boosts, and polar decompositions, which we
develop in this section.
\begin{proposition}\label{P:boostcomp}
For $\u,\v\in \R_c^3$, $B(\u)B(\v)=B(\u\oplus\v)h(\u,\v)$, where the right hand side is the polar decomposition of the
left hand side in the Lorentz  group $O^+(1,3)$.
\end{proposition}

\begin{proof}
 We calculate that
\begin{eqnarray*}
 B(\u)B(\v)&=& \left[\begin{matrix} \gamma_\u&\frac{\gamma _\u \u^T}{c}\\
\frac{\gamma_\u \u}{c} &I+\frac{\gamma_\u-1}{\vert\u\vert^2}\u\u^T
\end{matrix}\right]\left[\begin{matrix} \gamma_\v&\frac{\gamma _\v \v^T}{c}\\
\frac{\gamma_\v \v}{c} &\frac{\gamma_\v-1}{\vert\v\vert^2}\v\v^T
\end{matrix}\right]\\
&=&  \left[\begin{matrix}\gamma_\u \gamma_\v(1+\u^T\v/c^2)& *\\
\frac{\gamma_\u\gamma_\v \u}{c}+\frac{\gamma_\v \v}{c}+\frac{(\gamma_\u-1)\gamma_\v}{c\vert\u\vert^2}\u\u^T\v &*
\end{matrix}\right]
\end{eqnarray*}
We note from Proposition \ref{P:orthfact}  that the left column of the product must be the left column of the Lorentz boost
$P$ in the polar factorization $P\Theta$ of $B(\u)B(\v)$.   Since a Lorentz boost matrix  $P= \left[\begin{matrix} \gamma&\x^T\\
\x &S \end{matrix}\right]$ is  the Lorentz boost for the vector $\w=(c/\gamma)\x$, we conclude that the Lorentz boost $P$ in
the polar factorization of $B(\u)B(\v)$ is the Lorentz boost for the vector
\begin{eqnarray*}
 \frac{c}{\gamma_\u \gamma_\v(1+\frac{\u^T\v}{c^2})}\bigg(\frac{\gamma_\u\gamma_\v \u}{c} &+& \frac{\gamma_\v \v}{c}+\frac{(\gamma_\u-1)\gamma_\v}{c\vert\u\vert^2}\u\u^T\v\bigg)\\
 &=& \frac{1}{1+\frac{\u\cdot\v}{c^2}}(\u+\frac{1}{\gamma_\u}\v+\frac{\gamma_{\u} -1}{\gamma_\u\vert\u\vert^2}(\u\cdot \v)\u)\\
&=& \u\oplus \v.
\end{eqnarray*}
Thus $B(\u)B(\v)=P\Theta=B(\u\oplus\v)h(\u,\v)$, where we define $h(\u,\v)$ to be $\Theta$.
\end{proof}
\begin{problem}
Show that if $P= \left[\begin{matrix} \gamma& *\\ \x &* \end{matrix}\right]$ is a Lorentz boost, then $P=B(\w)$ for $\w=(c/\gamma)\x$
and $\gamma_\w=\gamma$.
\end{problem}
\begin{problem}\label{gammaform}
Argue from the proof of Proposition \ref{P:boostcomp} and the preceding problem that
\begin{equation}\label{E:gammaplus}
\gamma_{\u\oplus\v}=\gamma_\u \gamma_\v\Big(1+\frac{\u\cdot\v}{c^2}\Big).
\end{equation}
\end{problem}
\newpage

\section{Gyrogroups}
\begin{definition} (Groupoids or Magmas, and Automorphism Groups of Groupoids) 
A \emph{groupoid} or a \emph{magma}  is a nonempty set with a binary operation. An 
automorphism of the groupoid $(S, *)$ is a bijection of $S$ that respects the binary 
operation $*$ in $S$ . The set of all automorphisms of $(S, *)$ forms a group under
composition, denoted by Aut$(S,* )$. 
\end{definition}
An important subcategory of the category of groupoids is the category of loops.
\begin{definition} (Loops) A \emph{loop} is a magma $(S,\cdot)$ with an identity element in 
which each of the two equations $a \cdot x = b$ and $y\cdot a = b$ with unknowns $x$ and $y$ 
possesses a unique solution. As customary, we frequently denote the product
$a\cdot b$ by juxtaposition $ab$.  
\end{definition}

\begin{problem} Show that $(G,\cdot)$ is a group iff it is an associative loop.
\end{problem}

Being nonassociative, the Einstein velocity addition on the set of relativistically admissible 
velocities in the special theory of relativity is not a group operation. However,
it does possess group-like properties that  have been axiomatized by A.\ A.\ Ungar
as structures called ``gyrogroups," and studied in detail in his book 
\emph{Analytic Hyperbolic Geometry and Albert Einstein's Special Theory of Relativity}.
The gyrogroup concept abstracts both
Einstein's velocity addition and the corresponding Thomas
precession. The abstract Thomas precession is called the Thomas gyration
and suggests the prefix ``gyro'' for many of the concepts of the theory. 
 \begin{definition}  (Gyrogroups) The magma $(G,\oplus)$ is a 
\emph{gyrogroup} if its binary operation satisfies the following axioms. 
\begin{itemize}
\item[($\gamma$1)] There exists in $G$ some 
element, $0$, called a left identity, satisfying for all $a\in G$:
\begin{equation}
0 \oplus a = a \tag{Left Identity} 
\end{equation}
\item[($\gamma$2)] For each $a\in G$ there is an $x\in G$, called a left inverse of a, satisfying 
\begin{equation}
x\oplus a = 0 \tag{Left Inverse}
\end{equation} 
\item[($\gamma$3)] For any $a, b, z \in G$ there exists a unique element $\gyr[a, b]z \in G$ such 
that 
\begin{equation}
 a \oplus (b \oplus z) = (a \oplus b) \oplus \gyr[a, b] z \tag{Left Gyroassociative Law}
 \end{equation} 
\item[($\gamma$4)] If $\gyr[a, b]$ denotes the map $\gyr[a, b] : G \to G$ given 
by $z\mapsto  \gyr[a, b]z$ then 
\begin{equation}
\gyr[a, b] \in Aut(G, \oplus) \tag{Gyroautomorphism}
\end{equation}
and $\gyr[a, b]$ is called the \emph{Thomas gyration}, or the \emph{gyroautomorphism} of $G$, 
generated by $a, b \in G$. 
\item[($\gamma$5)] The gyroautomorphism $\gyr[a, b]$ generated by any 
$a, b \in G$ satisfies \begin{equation}
\gyr[a, b] = \gyr[a \oplus b, b] \tag{Left Loop Property} 
\end{equation}
\end{itemize}
\end{definition}

\begin{problem} For a gyrogroup $(G,\oplus)$ establish the following properties.
\begin{itemize}
\item[(1)] $a\oplus b=a\oplus c\Rightarrow b=c$ (left cancellation).
\item[(2)] $\gyr[0,a]=I$, the identity map on $G$.
\item[(3)] $\gyr[x,a]=I$ if $x\oplus a=0$.
\item[(4)] $\gyr[a,a]=I$.
\item[(5)] $a\oplus 0=a$, i.e., $0$ is an identity.
\item[(6)] There is only one left identity.
\item[(7)] Every left inverse is a right inverse.
\item[(8)] The left inverse, denoted $\ominus a$, is unique, and $\ominus(\ominus a)=a$.
\item[(9)] $\ominus a\oplus(a\oplus b)=b$.
\item[(10)] $\gyr[a,b]x=\ominus(a\oplus b)\oplus\big(a\oplus(b\oplus x)\big)$.
\item[(11)] $\gyr[a,b]0=0$. 
\item[(12)] $\gyr[a,b](\ominus x)=\ominus\gyr[a,b] x$.
\item[(13)] $\gyr[a,0]=\gyr[0,b]=I$.
\end{itemize}
\end{problem}

The preceding list of axioms is minimal in nature.  We typically work with the more extensive, but equivalent,
set of axioms.
 \begin{definition}\label{D:gyro2} (Gyrogroups-Alternative Definition) The magma $(G,\oplus)$ is a 
\emph{gyrogroup} if its binary operation satisfies the following axioms. 
\begin{itemize}
\item[(G1)] There exists in $G$ a unique identity element 
element $0$ satisfying for all $a\in G$:
\begin{equation}
0 \oplus a =a\oplus 0= a \tag{Identity} 
\end{equation}
\item[(G2)] For each $a\in G$, there exists a unique inverse  $\ominus a\in G$ satisfying 
\begin{equation}
\ominus a\oplus a = a\oplus (\ominus a)=0 \tag{Inverse}
\end{equation} 
\end{itemize}
For all $a,b\in G$, the map $\gyr[a, b]$ of $G$ into itself given by the 
equation 
\begin{equation}\label{E:GyrId}
\gyr[a, b]z = \ominus(a\oplus  b)\oplus(a\oplus  (b\oplus z)) 
\end{equation}
for all  $z \in G$, satisfies the following axioms: 
\begin{itemize}
\item[(G3)]  $\gyr[a, b]\in\mbox{Aut}(G, \oplus)$, the  gyroautomorphism group.
\item[(G4)] 
\begin{equation} a \oplus (b \oplus c) =(a \oplus b) \oplus \gyr[a, b]c \tag{Left Gyroassociative Law} 
\end{equation}
\begin{equation}
(a \oplus b) \oplus c = a \oplus (b \oplus \gyr[b, a]c) \tag{Right Gyroassociative Law} 
\end{equation}
\item[(G5)]
\begin{equation}
\gyr[a, b]= \gyr[a \oplus b, b] \tag{Left Loop Property}
\end{equation}
\begin{equation}
\gyr[a, b] = \gyr[a, b \oplus a] \tag{Right Loop Property}
\end{equation} 
\item[(G6)] A gyrogroup is called \emph{gyrocommutative} if it satisfies
\begin{equation}
a\oplus b=\gyr[a,b](b\oplus a)\tag{Gyrocommutative Law}
\end{equation}
\end{itemize}
\end{definition}

\begin{problem} Derive the second set of axioms (except gyrocommutativity) for a gyrogroup from the first set.
\end{problem}

\subsection{Involutive groups and gyrogroups}
We work in this section in the setting of an involutive group
$G$, a group equipped with an involutive automorphism $\tau$ such that
$\tau\circ\tau$ is the identity.  We
set $g^*=\tau(g^{-1})=(\tau(g))^{-1}$ and note that $g\mapsto
g^*:G\to G$ is an involutive antiautomorphism. Let
$$G^\tau:=\{x\in G: \tau(x)=x\}, \qquad P_G:= \{xx^*: x\in G\}
\subseteq G_\tau:=\{g\in G:g=g^*\}.$$ (That
$(xx^*)^*=x^{**}x^*=xx^*$ shows $P_G\subseteq G_\tau$.)

A subset $B$ of a group $G$ is called a \emph{twisted subgroup}
if the identity $e$ is in $B$, $B$ is closed under inversion, and
$xyx\in B$ whenever $x,y\in B$.
\begin{lemma}\label{L:Igs1}
The sets $P_G$ and $G_\tau$ are twisted subgroups.
\end{lemma}

\begin{proof}
Since $\tau(e^{-1})=\tau(e)=e$, we have $e\in P_G\subseteq G_\tau$.
Since $(g^*)^{-1}=\tau(g)=(g^{-1})^*$, we have that $g=g^*$ implies 
$(g^{-1})^*=g^{-1}$.  Also $(gg^*)^{-1}=(g^{-1})^*g^{-1}=
(g^{-1})^*(g^{-1})^{**}$.  Thus $G_\tau$ and $P_G$ are closed
under inversion.  

Let $gg^*,hh^*\in P_G$. Then $gg^*hh^*gg^*=gg^*h(gg^*h)^*\in P_G$.
Thus $P_G$ is a twisted subgroup.  A similar argument holds for $G_\tau$.  
\end{proof}

\begin{problem} Show that if $x,y\in G_\tau$, then $xyx\in G_\tau$.
\end{problem}

We recall other basic terminology.  Let $G$ be a group with subgroup
$H$. A subset $L$ of $G$ is said to be \emph{transversal to} $H$ 
if the identity $e\in L$ and $L$ intersects each coset $gH$ of $H$ 
in precisely one point.   One sees readily that a
subset $L$ containing $e$ is a transversal to $H$ if and only if the
map $(x,h)\mapsto xh:L\times H\to G$ is a bijection. In the case of
an involutive group $(G,\tau)$, if $L\subseteq \{g\in G:g=g^*\}$,
then the map $(x,k)\mapsto xk:L\times G^\tau\to G$ is called a
\emph{polar map}.  Hence $L$ containing $e$ is transversal to
$G^\tau$ if and only if the polar map is a bijection. If it is a
bijection, then the pair $(L,G^\tau)$ is called a \emph{polar
decomposition} for $(G,\tau)$.

We now come to the study of involutive groups with \emph{polar 
decomposition}. A twisted subgroup is \emph{uniquely $2$-divisible}
if each member of $P$ has a unique square root in $P$.
\begin{proposition}\label{IG:2}
Let $(G,\tau)$ be an involutive group, $P=\{gg^*\vert g\in G\}$.
The following are equivalent:
\begin{enumerate}
\item[(1)] $P$ is a uniquely $2$-divisible twisted subgroup.
\item[(2)] $P$ is transversal to $G^\tau$, i.e., the map
$(x,g)\mapsto xg:P\times G^\tau\to G$ is bijective.
\item[(3)] Every element $g\in G$ has a unique polar 
decomposition $g=xk\in PG^\tau$, $x\in P$, $k\in G^\tau$.
where $x=(gg^*)^{1/2}$.
\end{enumerate}
\end{proposition}

\begin{proof}
(1)$\Rightarrow$(3): If $g= x_1k_1=x_2k_2\in PG^\tau$, 
then $gg^*=(x_1)^2=(x_2)^2$. 
Hence $x_1=x_2$, and then $k_1=k_2$.  Thus factorizations, when
they exist, are unique.

For $g\in G$, set $x:=(gg^*)^{1/2}\in P$.  Choose $k\in G$ so that
$g=xk$.  We are finished if we show that $k\in G^\tau$.  We have
$$kk^*=(x^{-1}g)(x^{-1}g)^*=(gg^*)^{-1/2}gg^*(gg^*)^{-1/2}=e,$$
and thus $k^*=k^{-1}$, i.e., $k\in G^\tau$.  

(3)$\Rightarrow$(2):
Immediate.

(2)$\Rightarrow$(3): For $gg^*\in P$, let $g=xk\in PG^\tau$.  Then
$gg^*=xk(xk)^*=x^2$, so $x\in P$ is a square root of $gg^*$.  If
$y\in P$ were another, then one verifies that $y(y^{-1}g)$ would
give another decomposition of $g$, since
$$y^{-1}g=y^{-1}gg^*(g^*)^{-1}=y^{-1}y^2\tau(g)=y\tau(g)=
\tau(y^{-1}g).$$
\end{proof}

\begin{theorem}\label{T:transgyro}
Let $(G,\tau)$ be an  involutive group, $P=\{gg^*\vert g\in G\}$.
If $P$ is uniquely $2$-divisible, then $P$ is a gyrocommutative gyrogroup
for the operation $x\oplus y=(xy^2x)^{1/2}$.  The gyration automorphisms
given by  $\gyr[a,b]x=h(a,b)xh(a,b)^{-1}$, inner automorphism
by $h(a,b)$ where $ab=(a\oplus b)h(a,b)$, is the polar decomposition
of $ab$.
\end{theorem}

\begin{proof}  
By Lemma \ref{L:Igs1} $P$ is a twisted subgroup, then by Proposition \ref{IG:2} 
$P$ is transversal to $G^\tau$, and hence each element
of $G$ has a unique polar decomposition.  Furthermore, for each $a,b\in P$, the 
$P$-factor of the polar decomposition of $ab$ is given by
$$((ab)(ab)^*)^{1/2}=(ab^2a)^{1/2}=a\oplus b,$$
where the second equality is true by definition.  Thus
\begin{equation}\label{E:trans}
ab=(a\oplus b)h(a,b)\in PG^\tau,
\end{equation}
where $h(a,b)$ is defined to be the $G^\tau$-factor in the polar decomposition
of $ab$.

Directly from the definition $a\oplus b=(ab^2a)^{1/2}$, we conclude that
$e\oplus a=a=a\oplus e$ and $a\oplus a^{-1}=e=a^{-1}\oplus a$.  Hence
Axioms ($\gamma$1) and ($\gamma$2) are satisfied, and
$e=0$ and $a^{-1}=\ominus a$ in the gyrogroup terminology.

To verify ($\gamma$3), we have on the one hand that
\begin{eqnarray*}
(ab)c &=& (a\oplus b)h(a,b)c\\
&=&(a\oplus b)h(a,b)c(h(a,b)^{-1}h(a,b))\\
&=&  (a\oplus b)\gyr[a, b]c (h(a, b)) \\
&=& ((a\oplus b)\oplus \gyr[a, b]c)h(a \oplus b, \gyr[a, b]c)h(a, b).
\end{eqnarray*} 
and on the other hand that
\begin{eqnarray*}
a(bc)&=&a(b\oplus c)h(b,c)\\
 &=& a\oplus(b\oplus c) h(a,b\oplus c)h(b,c).
 \end{eqnarray*}
 Axiom ($\gamma$3) now follows from uniqueness of decomposition.

Let $k\in G^\tau$.  Then $kgg^*k^{-1}=kgk^{-1}kg^*k^{-1}=(kgk^*)(kgk^*)^*$, and
thus $kgg^*k^{-1}\in P$.  It follows that $P$ is invariant under inner automorphism by
any member of $G^\tau$.  Denote $kak^{-1}$ by $a^k$.  It is straightforward to verify that 
$(ab^2a)^k=a^k(b^k)^2a^k$ and hence that  $((ab^2a)^{1/2})^k=(a^k(b^k)^2a^k)^{1/2}$.
Since $h(a,b)\in G^\tau$ and $\gyr[a,b]x=x^{h(a,b)}$, we conclude that each $\gyr[a,b]$
is an automorphism of $(P,\oplus)$, and thus Axiom ($\gamma$4) is satisfied.

To establish that $(P,\oplus)$ is gyrocommutative, we consider the equations
\[
(ab)^*=\big((a\oplus b)h(a,b)\big)^*=h(a,b)^{-1}(a\oplus b) \]
and \[
(ab)^*=ba=(b\oplus a)h(b,a).\]
From these equations we conclude that
$$h(a,b)^{-1}(a\oplus b)=(b\oplus a)h(b,a)$$
and hence that
$$a\oplus b=h(a,b)(b\oplus a)h(a,b)^{-1}h(a,b)h(b,a)=\gyr[a,b](b\oplus a)h(a,b)h(b,a).$$
From uniqueness of the polar decomposition, we conclude 
$a\oplus b=\gyr[a,b](b\oplus a)$ (gyrocommutativity) and $h(a,b)^{-1}=h(b,a)$.

Finally, to establish ($\gamma$5), we observe 
$$b(ab)=b(a\oplus b)h(a,b)=(b\oplus(a\oplus b))h(b,a\oplus b)h(a,b).$$
It follows again from the uniqueness of the polar decomposition that $
h(b,a\oplus b)^{-1}=h(a,b)$. From the last of the preceding paragraph 
we see that $h(b,a\oplus b)^{-1}=h(a\oplus b,b)$, and we conclude
$$h(a\oplus b,b)=h(b,a\oplus b)^{-1}=h(a,b).$$
We then have directly from the definition of the gyroautomorphisms that
$\gyr[a\oplus b,b]=\gyr[a,b]$.
\end{proof}

\begin{Remark}
 We remark that there is a converse to the preceding theorem, namely that every uniquely $2$-divisible
gyrocommutative gyrogroup can be realized (up to isomorphism) as $(P,\oplus)$ for some
involutive group satisfying the hypotheses of the preceding theorem.
\end{Remark}

There is a weaker converse for general gyrogroups, which we present in the next problem.
\begin{problem}\label{P:semprod}
 Let $(G,\oplus)$ be a gyrogroup and let $\A$ be a subgroup of its automorphism group containing all
automorphisms $\gyr[a,b]$, $a,b\in G$.  Then $G\times\A$ is a group with respect to the operation
\begin{equation}
 (a,A)(b,B)=(a+Ab,\gyr[a,Ab]A B).
\end{equation}
\end{problem}

\begin{problem}
 Show that  for $a\in G$, the map $na \mapsto (a,I)^n$ is an isomorphism from the subgyrogroup of $G$
generated by $a$ to the subgroup of $G\times \A$ generated by $(a,I)$, where $(n+1)a=na\oplus a$.  
In particular, a gyrogroup is power associative.
\end{problem}

\subsection{The Einstein gyrogroup}
The main goal of this section is to show that the set of admissible velocities $\R_c^3$ endowed with the Einstein addition
is a gyrocommutative gyrogroup.
\begin{lemma} \label{L:Logrp}
 The Lorentz group endowed with the involution for which $A^*=A^T$ satisfies the hypotheses of Theorem
\ref{T:transgyro}  with $P$ equal to the set of Lorentz boosts and $O^+(1,3)^\tau$ equal to the subgroup
of orthogonal Lorentz transformations.
\end{lemma}

\begin{proof}  In the Lorentz  group $O^+(1,3)$ it is straight forward to verify that $\tau(A)=(A^{-1})^T$ is an involution
 with $A^*=A^T$.  We verify that $P=\{AA^T:A\in O^+(1,3)\}$ is the set of Lorentz boosts.  On the one hand, for any $A\in O(^+(1,3)$,
we have directly that $AA^T$ is a positive definite Lorentz transformation, hence a Lorentz boost by Proposition
\ref{P:eqboost}.  Conversely let $A$ be a Lorentz boost, say $B(\v)$.   It follows from Problem \ref{oneparam} that there exists
a velocity $\u$ such that $\u\oplus \u=\v$.  By Proposition \ref{P:Ldecomp}, we obtain $B(\u)B(\u)=B(\u\oplus \u)h(\u,\u))$.  But
$B(\u)B(\u)=B(\u)B(\u)^*$ is in $P$ already, so the factor $h(\u,\u)$ is the identity, and $$B(\v)=B(\u\oplus \u)=
B(\u)B(\u)=B(\u)B(\u)^*\in P.$$
If $(U^T)^{-1}=\tau(U)=U$, then taking inverses we obtain $U^T=U^{-1}$ if and only if $U$ is orthogonal.  Hence 
the Lorentz transformations fixed by $\tau$ are precisely the orthogonal ones.
\end{proof}

\begin{theorem}\label{T:eingyro}
 The correspondence $\v\mapsto B(\v)$ defines an isomorphism between the gyrocommutative gyrogroups
 $(\R_c^3,\oplus)$ of admissible velocities under Einstein velocity addition and the set $P=\{B(\u):\u\in\R_c^3\}$ of
 Lorentz boosts under the operation $B(\u)\oplus B(\v)=(B(\u)B(\v)^2B(\u))^{1/2}$.
\end{theorem}

\begin{proof}
By Proposition \ref{P:boostcomp} for $\u,\v\in \R_c^3$, 
$B(\u)B(\v)=B(\u\oplus\v)h(\u,\v)$, where the right hand side is the polar decomposition of the
left hand side in the Lorentz  group $O^+(1,3)$.  By Lemma \ref{L:Logrp} and the first paragraph of the proof
of Theorem \ref{T:transgyro}, we have  $B(\u)B(\v)=(B(\u)\oplus B(\v))h(B(\u),B(\v))$.  By uniqueness of the 
polar decomposition in $O^+(1,3)$, we have $B(\u)\oplus B(\v)=B(\u\oplus\v)$, which shows that $\u\to B(\u)$ is
an isomorphism (since essentially by definition it is a bijection).  It is straightforward to verify that an isomorphism of
groupoids (magmas) preserves all the properties of Definition \ref{D:gyro2}, and hence one of the systems is a 
gyrocommutative gyrogroup if and only if the other is.
\end{proof}

\subsection{Basic theory of gyrogroups}
We assume throughout this section the axioms of Definition \ref{D:gyro2} for a gyrogroup.  We first list some
basic properties of the gyrations; see \cite[Chapter 2]{Ung}.
\begin{proposition}\label{P:gyrprops}
 Let $(G,\oplus)$ be a gyrogroup.  Then for all $a,b\in G$ 
the gyrations satisfy the following properties:
\begin{itemize}
\item[(i)] $\gyr[\ominus a,\ominus b]=\gyr[a,b]$;
\item[(ii)] $\gyr[b,a]=\gyr^{-1}[a,b]$, the inverse of $\gyr[a,b]$;
\item[(iii)] $\gyr[a\oplus b,\ominus a]=\gyr[a,b]$;
\item[(iv)] $\gyr[na,ma]=I$ for all integers $m,n$.
\end{itemize}
\end{proposition}

\begin{problem}\label{P:switch}
Using Problem \ref{P:semprod} and the fact that $(a,A)^{-1}=(A^{-1}(\ominus a),A^{-1})$, invert both sides of the equation
$(a,I)(b,I)=(a\oplus b,\gyr[a,b])$ to show that $\gyr[\ominus b,\ominus a]=\gyr^{-1}[a,b]$ and $\ominus(a\oplus b)=\gyr[a,b]
(\ominus b\ominus a)$.
\end{problem}

A magma is called \emph{left power alternative}  if  for all $a,b$ and all integers $m,n$, $$ma\oplus(na\oplus b)=(m+n)a\oplus b.$$
\begin{corollary}
 A gyrogroup is left power alternative.
\end{corollary}

\begin{problem}
 Use Proposition \ref{P:gyrprops}(iv)  to prove the preceding corollary.
\end{problem}

\begin{proposition}\label{P:lBol}
A gyrogroup satisfies the left Bol identity
\begin{equation}\label{E:lBol}
a\oplus(b\oplus(a\oplus c))=(a\oplus(b\oplus a))\oplus c.
\end{equation}
\end{proposition}
 
\begin{proof}
We have
\begin{eqnarray*}
a\oplus(b\oplus(a\oplus c))&=& a\oplus ((b\oplus  a)\oplus \gyr[b,a]c)\\
&=& (a\oplus (b \oplus a))\oplus \gyr[a,b\oplus a]\gyr[b,a] c.
\end{eqnarray*}
Noting from Proposition \ref{P:gyrprops} that $\gyr[a,b\oplus a]\gyr[b,a]
=\gyr[a,b]\gyr[b,a]=I$, we obtain the result.
\end{proof}

In the  next two propositions we show that a gyrogroup is a loop.
\begin{proposition}\label{P:Lloop}
In a gyrogroup $(G,\oplus)$, the equation $a\oplus x=b$ in unknown $x$ has the unique 
solution $x=\ominus a\oplus b$.
\end{proposition}

\begin{problem} Prove Proposition \ref{P:Lloop}
\end{problem}

\begin{definition}\label{D:coop}
In a gyrogroup $(G,\oplus)$, we define the \emph{coaddition} $\boxplus$ by
$$a\boxplus b=a\oplus \gyr[a,\ominus b]b.$$
We set 
$$a\boxminus b:=a\boxplus (\ominus b)=a\oplus\gyr[a,b](\ominus b)=a\ominus\gyr[a,b]b.$$
\end{definition}

\begin{proposition}\label{P:Rloop}
In a gyrogroup $(G,\oplus)$ the equation $x\oplus a=b$ has the unique solution 
$x=b\boxminus a$.
\end{proposition}

\begin{proof}
If $x\oplus a=b$, then 
\begin{eqnarray*}
x&=&x\oplus (a\ominus a)\\
&=& (x\oplus a)\oplus\gyr[x,a](\ominus a)\\
&=& (x\oplus a)\ominus\gyr[x,a]a\\
&=& (x\oplus a)\ominus\gyr[x+a,a]a\\
&=& b\ominus \gyr[b,a]a\\
&=& b\boxminus a.
\end{eqnarray*}
Conversely we first note that
\begin{eqnarray*}
b\boxminus a&=&b\boxplus (\ominus a)\\
&=& b\oplus\gyr[b,a](\ominus a)\\
&=& \lambda_b\lambda_{\ominus(b\oplus a)}\lambda_b\lambda_a(\ominus a)\\
&=& \lambda_b\lambda_{\ominus(b\oplus a)}\lambda_b(0)\\
&=& b\oplus(\ominus(b\oplus a)\oplus b).
\end{eqnarray*}
Therefore
$$(b\boxminus a)\oplus a= (b\oplus(\ominus(b\oplus a)\oplus b))\oplus a
=b\oplus(\ominus(b\oplus a)\oplus (b\oplus a))=b.$$
where the second equality follows from the left Bol identity.
\end{proof}

\begin{problem}  Show for a gyrogroup $(G,\oplus)$ that 
 $(G,\boxplus)$ is a loop with the same identity and inverses
as $(G,\oplus)$.
\end{problem}

We consider a basic alternative characterization of gyrocommutative gyrogroups.
\begin{proposition} \label{P:commute}
A gyrogroup $(G,\oplus)$ is gyrocommutative if and only if it satisfies the
\emph{automorphic inverse property}
$$\ominus(a\oplus b)=\ominus a\ominus b.$$
\end{proposition}

\begin{proof}
Assume that $G$ is gyrocommutative.  By Problem \ref{P:switch}
$$\ominus(a\oplus b)=\gyr[a,b](\ominus b\oplus(\ominus a))
=\gyr[a,b]\gyr[\ominus b,\ominus a](\ominus a\ominus b)=\ominus a\ominus b.$$

Conversely, assume that $G$ satisfies the automorphic inverse property.  Then
using Proposition \ref{P:gyrprops}(i) and gyroassociativity, we obtain
\begin{eqnarray*}
\ominus(a\oplus b)\oplus\gyr[a,b](b\oplus a)&=&(\ominus a\oplus(\ominus b))
\oplus\gyr[\ominus a,\ominus b](b\oplus a)\\
&=& \ominus a\oplus(\ominus b\oplus(b\oplus a))\\
&=& \ominus a\oplus a=0.
\end{eqnarray*}
It follows that $\gyr[a,b](b\oplus a)$ is the inverse of $\ominus(a\oplus b)$, and
hence must equal $a\oplus b$.

\end{proof}

\section{The Einstein gyrovector space}
In this section we enrich the structure of the Einstein velocity gyrogroup.  Our ultimate
goal is to equip it with enough structure to carry out an analytic hyperbolic geometry.

\subsection{Gyrovector spaces}
In this section let $(G,\oplus)$ be a gyrocommutative gyrogroup.   
\begin{definition}
A \emph{gyrovector space} consists of a gyrocommutative group $(G,\oplus)$ such that for 
every $0\ne x\in G$, there exists a unique injective homomorphism $\alpha_x$ of $(\R,+)$ into $(G,\oplus)$
such that $\alpha_x(1)=x$.  In this case we define a \emph{scalar multiplication}  from $\R\times G$ to
$G$ by $r.x=\alpha_x(r)$.  We sometimes write $x.r$ for $r.x$  to avoid proliferation of parentheses.
\end{definition}

\begin{lemma}\label{L:scalprop}
 In a gyrovector space scalar multiplication satisfies
\begin{itemize}
 \item[(i)] $1.x=x$ and $(-1).x=\ominus x$, $0.x=0$;
\item[(ii)] $(s+t).x=s.x\oplus t.x$;
\item[(iii)] $s.t.x=(st).x$
\item[(iv)] $m.x=\bigoplus_{i=1}^m x$ for any positive integer $m$.
\end{itemize}
\end{lemma}

\begin{proof}  For (iii) let $\alpha_x:\R\to G$ be a homomorphism with $\alpha_x(1)=x$.  Then $t.x=\alpha_x(t)$.
 Define $\beta:\R\to G$ by $\beta(r)=\alpha_x(rt)$. Then $\beta$ is a homomorphism (since $\alpha_x$ is) and
$\beta(1)=\alpha_x(t)=t.x$.  Thus by definition $\beta(s)=s.t.x$.  But also $\beta(s)=\alpha_x(st)=(st).x$, and hence
the two are equal.
\end{proof}

\begin{problem}
 Verify the other conclusions of the preceding lemma.
\end{problem}

\begin{problem}
 Suppose that one is given a scalar multiplication $(r,x)\mapsto r.x$ satisfying the conditions of the preceding
lemma.  Show that  if one defines $\alpha_x(t)=t.x$ for each nonzero $x$, then  one obtains a gyrovector space.
What  is a minimal set of the properties in Lemma \ref{L:scalprop} needed to derive this result?
\end{problem}

\begin{definition} A \emph{topological  gyrovector space} is a gyrovector space $G$ equipped with a Hausdorff topology
such that  $\oplus:G\times G\to G$ and $t.x:\R\times G\to G$ are continuous.
\end{definition}

\begin{definition}
An \emph{exponential function} for a topological gyrovector space $G$ is a homeomorphism $\exp:V\to G$ 
from a real topological vector space $V$ to $G$ such that the restriction to any one-dimensional subspace is
an additive homomorphism into $G$.
\end{definition}

\begin{lemma}\label{L:gyrovect}
The continuity of the scalar multiplication $t.x:\R\times G\to G$ in a topological gyrovector space follows 
from the existence of an exponential function.
\end{lemma}

\begin{proof}
The scalar multiplication can be written as the continuous composition $(t,x)\mapsto \exp(t.\log x)$.
\end{proof}

\subsection{The exponential for gyroboosts and admissible velocities}
We begin with a lemma which we have essentially proved, but never formally stated.
\begin{lemma}\label{L:P}
The following are equivalent for $A\in O^+(1,3)$:
\begin{enumerate}
\item $A$ is a Lorentz boost.
\item $A$ is positive definite.
\item $A=BB^*$ for some $B\in O^+(1,3)$.
\end{enumerate}
\end{lemma}

\begin{proof}
That (3) implies (2) is immediate and (2) implies (1) is the content of Proposition  \ref{P:eqboost}.
Suppose that $A=B(\u)$ is a Lorentz boost.  From the fact that one dimensional subspaces intersect
$\R_c^3$ in one-dimensional subgroups isomorphic to $(\R,+)$, we have that 
$\u=t\u\oplus t\u$ for some $t$.  Since $B(t\u)B(t\u)=B(t\u)B(t\u)^T\in P$, we have that the
$O(3)$ component of the polar decomposition of $B(t\u)B(t\u)$ is the identity and hence
$$B(t\u)B(t\u)=B(t\u\oplus t\u)=B(\u).$$
Hence $B(\u)\in P$.
\end{proof}

Let $P$ be the subset of the Lorentzian group $O^+(1,3)$ consisting of all Lorentzian boosts.
The Lie algebra $\mathfrak{o}^+(1,3)=\{X\in M_4(\R): \forall t\in \R,\ \exp(tX)\in O^+(1,3)\}$ of $O^+(1,3)$
is computed in the standard way for Lie groups defined by preserving a bilinear form and is given by
$$\mathfrak{o}^+(1,3)=\{X\in M_4(\R): I_{1,3}X+XI_{1,3}=0\}.$$

Our first goal is to compute the tangent space  $\mathfrak p$ of $P$.

\begin{lemma}\label{L;frakp}
 The tangent space $\p=\{X\in\frak{o}^+(1,3): \exp(tX)\in P\mbox{ for all }t\in\R\}$ of $P$ is given by
$$\p=\{X\in\frak{o}^+(1,3): X=\left[\begin{matrix} 0& \u^T\\ \u& 0\end{matrix}\right]\mbox{ for some }\u\in\R^3\}.$$
\end{lemma}

\begin{proof}  If $\exp(tA)=e^{tA}$ is symmetric for all $t$, then $$A=\frac{d}{dt} e^{tA}\vert_{t=0}=\lim_{t\to 0} \frac{e^{tA}-I}{t}$$
 is symmetric.  It is straightforward to verify that  the conditions of symmetry and $ I_{1,3}X+XI_{1,3}=0$ imply that $X$ must
be of the form specified in the lemma.

For the converse direction, consider $X=\left[\begin{matrix} 0& \u^T\\ \u& 0\end{matrix}\right]$.  By direct computation one
verifies that $X^{2n}=\vert\u\vert^{2n-2}\left[\begin{matrix} \u^T\u& 0\\ 0& \u\u^T\end{matrix}\right]$ and $X^{2n+1}=\vert \u\vert^{2n}X$.
It follows that 
$$\exp(tX)=\left[\begin{matrix} \cosh(t\vert\u\vert) & \frac{\sinh(t\vert\u\vert)}{\vert\u\vert}\u^T\\ \frac{\sinh(t\vert\u\vert)}{\vert\u\vert}\u
& I+\frac{\cosh(t\vert\u\vert)-1}{\vert\u\vert^2}\u\u^T\end{matrix}\right].$$
One verifies directly that the preceding matrix satisfies the conditions of equation \ref{E:Sdef}, and hence is a Lorentz boost by 
Proposition \ref{P:eqboost}.
\end{proof}

\begin{problem}  Verify in detail that the powers of $X$ and $\exp(tX)$ are indeed as asserted in the previous proof.  
Verify in detail that $\exp(tX)$ is a Lorentz boost.  
\end{problem}

\begin{proposition} The exponential map from $\p$ to $(P,\oplus)$ is an exponential map from $\p$ to the topological 
gyrovector space of Lorentz boosts.  
\end{proposition}

\begin{proof}  We have seen in Section 3.2 that $\exp:\mbox{Sym}\to \mbox{Sym}^{>0}$ is a homeomorphism and a group
homomorphism on each one-dimensional subspace of Sym.  Since $P$ is the intersection of the set of $O^+(1,3)$ and the 
positive definite matrices (Lemma \ref{L:P}) and the latter two sets are closed in the general linear group $GL_4(\R)$, it follows that
$\exp^{-1}(P)$ is closed in Sym.  It follows easily from Problem \ref{oneparam} that each element of $(\R_c^3,\oplus)$ has an $n^{th}$-root
for each positive integer $n$.  By Theorem \ref{T:eingyro} the isomorphic gyrogroup of Lorentz boosts must have $n^{th}$-roots for each
element.  But for a positive definite element $A$  these roots are unique and are given by $\exp((1/n)(\log A))$.  Hence $B_n=(1/n)(\log A))$
is in $\exp^{-1}(P)$  for each $n$.  Then for each integer $m$, $\exp((m/n)B_n)=(B^n)^m\in P$,  since the product is positive definite
and in $O^+(1,3)$.  It follows that $\exp^{-1}(P)$ contains a dense subset of $\R\log A$, and by its closeness must therefore
contain $\R\log A$.  It follows that $\R\log A\subseteq \p$, in particular $\log A\in \p$.  Thus $\log P\subseteq \p$, or applying
$\exp$, we obtain $P\subseteq \exp \p$.  The reverse inclusion is immediate.  Thus $\exp:\p\to P$ is a homeomorphism
that is  group homomorphism on one-dimensional subspaces into the multiplicative structure of $P$.    It follows from
Theorem \ref{T:eingyro}  that the multiplication agrees with gyroaddition on commutative subgroups of $P$, in particular
on the image of one-parameter subgroups.  Hence $\exp$ restricted to any one-dimensional subspace of $\p$
is a homomorphism into $(P,\oplus)$.  

Since gyroaddition in $P$ is multiplication followed by projection into the $P$-factor of the product, it is continuous.
By Lemma \ref{L:gyrovect} the scalar multiplication is continuous.

\end{proof}

\begin{problem}  The commutator product of two $n\times n$-matrices $X$ and $Y$ is defined by $[X,Y]=XY-YX$.  
 Show that $\frak{o}^+(1,3)$ is closed under commutator product (and is hence a \emph{Lie algebra}) and
 $\p$ is closed under the \emph{triple product} $\langle X,Y,Z\rangle:=[X,[Y,Z][$.
\end{problem}

Let $(\R_c^3,\oplus)$ be the gyrogroup of admissible velocities.  We define an exponential function 
$\exp:\R^3\to \R_c^3$ by $\exp(\u):=c\tanh(\vert \u\vert)(\u/\vert\u\vert)$.  

\begin{proposition}\label{P:expAV}
 Each of the maps in the following diagram is a diffeomorphism (smooth homeomorphism) and the diagram commutes:
\[\begin{diagram}
\R^3    &\rTo^{\beta}      &\p \\
 \dTo^{\exp}   &   &\dTo_{\exp} \\
 \R_c^3               & \rTo^{B} & P.
\end{diagram}\]
In the diagram $B(\u)$ is the Lorentz boost for $\u$ and $\beta(\u)=\left[\begin{matrix} 0& \u^T\\ \u& 0\end{matrix}\right]$.
\end{proposition}

\begin{proof}
 The horizontal maps are coordinatewise smooth, hence smooth.  The matrix exponential function is a power series map, hence
smooth.  From commutativity of the diagram and the bijectivity of each map, it follows that the exponential on $\R^3$ is smooth.
\end{proof}

\begin{problem}  Verify that all the maps are bijections and that the diagram commutes.
 \end{problem}

From the preceding proposition we easily obtain the
\begin{corollary}\label{C:expAV}
 The map $\exp:\R^3\to \R_c^3$ defined above is an exponential map for
$(\R_c^3,\oplus)$ and hence $(\R_c^3,\oplus)$ is a topological gyrovector space.
\end{corollary}

\subsection{Rooted vectors and gyrolines}  In this section we consider some basic notions
of the vector geometry of gyrovector spaces.   We begin by adding two additional axioms to our notion of
a gyrovector vector space.
\begin{definition} \label{gyrovec}A \emph{gyrovector space} is a gyrocommutative gyrogroup $V$ equipped with 
 a scalar multiplication $(t,x)\mapsto t.x$ from $\R\times V\to V$ that satisfies:
\begin{enumerate}
 \item $1.\x=\x$;  $ -1.\x=\ominus \x$; $0.\x=t.0=0$;
\item $(s+t).\x=s.\x\oplus t.\x$;
\item $s.t.\x=(st).\x.$;
\item $\gyr[\a,\b](s.x)= s.\gyr[\a,\b]x$;
\item $\gyr[s.\a,t.\a]=I$.
\end{enumerate}
\end{definition}

 In a gyrovector space $(V,\oplus,.)$ it is 
convenient to think of members of $V$ in two distinct ways: as (geometric) points $P,Q,R$
and as vectors $\u$, $\v$, $\w$ emanating from the origin $0$.  A rooted gyrovector is
viewed as a vector emanating from other points in addition to the origin.  More formally, 
a \emph{rooted gyrovector} is a pair $(P,\v)\in V\times V$.  We can alternatively write a rooted
gyrovector in the equivalent form $\oa{PQ}$, where $\v=\ominus P\oplus Q$.  
The rooted vector $\oa{PQ}$ has \emph{head} $P$ and \emph{tail} $Q$.  
%For a rooted vector $\oa{PQ}$, let $\a=P$ and $\b=Q$.
\begin{lemma}\label{gyrotrans1}
In a gyrogroup $(G,\oplus)$, 
$$\ominus(a\oplus b)\oplus(a\oplus c)=\gyr[a,b](\ominus b\oplus c).$$
\end{lemma}

\begin{proof}
The proof is a straightforward application of the fact that $\gyr[a,b]=\lambda_{\ominus(a\oplus b)}
\lambda_a\lambda_b$.
\end{proof}

\begin{proposition}\label{P:gve}
In a gyrovector space $V$, the following are equivalent for $P,Q,P',Q'\in V$.
\begin{itemize}
\item[(1)] For some $\v\in V$, $Q=P\oplus \v$ and $Q'=P'\oplus \v$.
\item[(2)] $\ominus P\oplus Q=\ominus P'\oplus Q'$.
\item[(3)] For some $\u\in V$, $P'=\gyr[P,\u](\u\oplus P)$ and $Q'=\gyr[P,\u](\u\oplus Q)$.
\end{itemize}
In case $(3)$ the vector $\u$ is unique and given by $\u=\ominus P\oplus P'$.
\end{proposition}

\begin{proof} (1)$\Leftrightarrow$(2): From $Q=P\oplus\v$, we deduce that $\ominus P\oplus Q=\v$ and similarly 
$\ominus P'\oplus Q'=\v$. In the converse case simply set $\v=\ominus P\oplus Q=\ominus P'\oplus Q'$.

(2)$\Rightarrow$(3): Set $\u=\ominus P\oplus P'$. Then
$$\gyr[P,\u](\u\oplus P)=P\oplus \u=P\oplus(\ominus P\oplus P')=P',$$
where the first equality follows from gyrocommutativity.  Since $\ominus P\oplus Q=\ominus P'\oplus Q'$,
we have
\begin{eqnarray*}
Q'&=&P'\oplus(\ominus P\oplus Q)=\gyr[P,\u](\u\oplus P)\oplus (\ominus P\oplus Q)\\
&=& \gyr[P,\u]\bigl( (\u\oplus P)\oplus(\gyr[\u,P](\ominus P\oplus Q))\bigr)\\
&=&\gyr[P,\u]\bigl(\u\oplus(P\oplus(\ominus P\oplus Q))\bigr)\\
&=& \gyr[P,\u](\u\oplus Q)
\end{eqnarray*}

(3)$\Rightarrow$(2): We have by gyrocommutativity, $P'=\gyr[P,\u](\u\oplus P)$ 
and $Q'=\gyr[P,\u](\u \oplus Q)$, so by Lemma \ref{gyrotrans1}
$$\gyr[\u,P] (\ominus P\oplus Q)=\ominus(\u\oplus P)\oplus(\u\oplus Q).$$
Applying $\gyr[P,\u]$ to both sides yields
$$\ominus P\oplus Q=\gyr[P,\u](\ominus(\u\oplus P))\oplus\gyr[P,\u](\u\oplus Q)=
\ominus P'\oplus Q'.$$

The uniqueness of $\u$ in condition (3) follows from the fact that
$$\ominus P\oplus P'=\ominus P\oplus\gyr[P,\u](\u\oplus P)=
\ominus P\oplus (P\oplus \u)=\u.$$
\end{proof}

Two rooted gyrovectors  $\overrightarrow{PQ}$ and $\overrightarrow{P'Q'}$ are \emph{equivalent}
if they satisfy the equivalent conditions of Proposition \ref{P:gve}. 

\begin{definition}  A \emph{gyroline} in a gyrovector space $V$ is a set of the form
$$\{P\oplus t.\v: t\in \R,\ P,\v\in V,\ \v\ne 0\}.$$
The specific gyroline is called the gyroline through $P$ in direction $\v$.
The map $t\mapsto P\oplus t.\v$ is a \emph{linear parameterization} of the
gyroline.  
\end{definition}

We list some elementary properties of gyrolines.

\begin{lemma}\label{L:2pt}
Given two distinct points $P$ and $Q$ in a gyrovector space $V$, there exists a 
gyroline containing the two points.  The parameterization $t\mapsto P\oplus t.(\ominus P
\oplus Q)$ is a  gyrolinear parameterization taking on the value $P$ at $0$ and
$Q$ at $1$.  
\end{lemma}

\begin{proof} By definition $\{P\oplus t.(\ominus P\oplus Q): t\in\R\}$ is a gyroline,
and one sees directly that the given parametrization is gyrolinear and takes on the
values $P$ and $Q$ at $0$ and $1$ resp.
\end{proof}

The notion of a gyroline allows a geometric visualization of a gyrovector
$\oa{PQ}$.  Namely we consider the segment of the gyroline determined by
$P$ and $Q$ that lies between $P$ and $Q$, namely $\{P\oplus  t.(\ominus P\oplus Q): 
0\leq t \leq 1\}$, directed in the direction from $P$ to $Q$.

\begin{lemma}\label{L:rtline}
 The left translation of a gyroline is a gyroline.
\end{lemma}

\begin{proof}   We note that $P\oplus (Q\oplus t.\v)=(P\oplus Q)\oplus \gyr[P,Q](t.\v)=
(P\oplus Q)\oplus t.\gyr[P,Q]\v.$. Thus the left translation by $P$ of the gyroline through $Q$ in the 
direction $\v$ is the gyroline through $P+Q$ in the direction $\gyr[P,Q]\v$. 
\end{proof}

\begin{lemma} \label{L:gyrorig}
 Any gyroline $\ell$ through $0$ has a parametrization of the form $\alpha_\v(t)=t.\v$ for some $\v\ne 0$ in $\ell$, 
which is an injective gyrovector space homomorphism.
\end{lemma}

\begin{proof}
Let  $\ell=\{P+t.\v:t\in\R\}$ contain $0$, i.e., $0=P\oplus r.\v$ for some $r\in \R$.  Then $r.\v=\ominus P$ or
$P=(-r).\v$.  Set $\alpha_\v(t)=t.\v$.  Then
$$\alpha_v (t)=t.\v=(-r).\v\oplus (t+r).\v=P\oplus (t+r)\v.$$
is a gyrolinear parameterization of $\ell$, a gyrovector space homomorphism by properties (2) and (3)
of Definition \ref{gyrovec}, and injective by property (1).
\end{proof}

\begin{corollary}\label{C:uniqueline}
Given two distinct points of a gyrovector space, there is a unique gyroline containing the two points.
\end{corollary}

\begin{proof} Let $P,Q$ be distinct points.  By Lemma \ref{L:2pt} there exists a gyroline $\ell$ containing
the two.   Consider the special case that $Q=0$ (hence $P\ne 0$).  By Lemma \ref{L:gyrorig} there
exists $\v\ne 0$ such that  $\ell=\{t.\v:t\in\R\}$.  Let $P=s.\v$; then $s\ne 0$.  Since $t.v=(t/s).P$ for all 
$t\in\R$, It follows that $\ell=\{t.P:t\in\R\}$.  Since $\ell$ was an arbitrary gyroline containing $0$ and $P$,
it follows that $\{t.P:t\in\R\}$ is the unique gyroline containing $0$ and $P$.  One can now use
Lemma \ref{L:rtline} to argue that there is a unique gyroline through any two distinct points.
\end{proof}

\begin{remark}\label{R:uniqueline}
It follows from Corollary \ref{C:uniqueline} and Lemma \ref{L:2pt} that 
$t\mapsto P\oplus t.(\ominus P\oplus Q)$ is a  gyrolinear parameterization 
of the unique gyroline through any two distinct points $P,Q$.
\end{remark}

\subsection{Gyrovector spaces with inner product}
To talk about length of gyrovectors and angles between gyrovectors, we introduce an inner product.
\begin{definition}\label{D:innergyro}
A \emph{real inner product gyrovector space} consists of three components:
\begin{itemize}
\item[(1)] A gyrovector space $(G,\oplus,.)$ defined on some open ball of a real inner product vector space.
\item[(2)]  The set  $\Vert G\Vert:=\{\pm\Vert \v\Vert : \v\in G\}$ is equipped with a gyroaddition and scalar 
multiplication making it a gyrovector space.  Here the norm is the one induced by the
inner product.
\item[(3)] The gyrovector space structure is connected to the inner product through the following laws:\\
$(i) ~\gyr[\u,\v]\a\cdot\gyr[\u,\v]\b=\a\cdot\b$, i.e., gyrations preserve inner product.\\
$(ii)~\frac{\vert r\vert.\a}{\Vert r.\a\Vert}=\frac{\a}{\Vert \a\Vert}$.\\
$(iii)~\Vert r.\a\Vert=\vert r\vert.\Vert\a\Vert.$\\
$(iv)~\Vert \a\oplus \b\Vert\leq \Vert\a\Vert\oplus\Vert\b\Vert$.
\end{itemize}
\end{definition}

\begin{remark}\label{zero=0}
The zero of the gyrogroup $G$ is equal to $0$ in the real inner product vector space.
\end{remark}

\begin{proof}
We have 
$$0=\Vert 0\Vert=\vert -1\vert.\Vert 0\Vert=\Vert -1.0\Vert=\Vert \ominus 0\Vert.$$
It follows that $\ominus 0=0$.  Adding $0$ to both sides , we obtain
$Z=0\oplus 0=2.0$, where $Z$ is the gyrogroup additive identity.  Multiplying
both sides by $1/2$ yields $Z=(1/2).Z=0$.
\end{proof}

\subsection{The gyrodistance}
We assume in this section we are working in an inner product gyrovector space $G$.
\begin{definition}
The \emph{gyrodistance} function $d_\oplus(\a,\b)$ is defined by 
$$ d_\oplus(\a,\b)=\Vert \ominus \a\oplus \b\Vert.$$
\end{definition}

\begin{proposition}\label{P:gyrodist}
The gyyrodistance function satisfies the standard axioms for a metric, with addition in the
triangle inequality replaced by gyroaddtion.
\end{proposition}

\begin{proof}
 From Remark \ref{zero=0}, we have that $\Vert \v \Vert=0$ if and only
if  $\v=0$, the gyroaddition identity.  If $\Vert\ominus \a\oplus\b\Vert=0$,
then $\ominus \a\oplus \b=0$ and hence $\a=\b$.

Since a gyrovector space is gyrocommutative, we have 
$$\ominus \a\oplus \b=\ominus(\a\ominus \b)=\ominus\gyr[\a,\ominus \b](\ominus\b\oplus\a).$$
Since $\Vert \ominus \c\Vert=\Vert(-1).\c\Vert=\vert -1\vert. \Vert\c\Vert=\Vert\c\Vert$, we
have 
$$\Vert \ominus \a\oplus \b\Vert=\Vert \ominus\gyr[\a,\ominus \b](\ominus\b\oplus\a)\Vert
=\Vert\gyr[\a,\ominus \b](\ominus \b\oplus\a)\Vert=\Vert\ominus\b \oplus \a\Vert.$$
where the last equality follows from the fact that $\gyr[\a,\ominus \b]$ preserves the
norm since it preserves the inner product. It follows that $d_\oplus(\a,\b)=
d_\oplus(\b,\a)$.

By Lemma \ref{gyrotrans1}
$$\ominus(\ominus\a\oplus \b)\oplus(\ominus \a\oplus \c)=\gyr[\ominus \a,\b](\ominus\b\oplus \c),$$
and hence
$$\ominus \a\oplus\c=(\ominus \a\oplus\b)\oplus\gyr[\ominus \a,\b](\ominus\b\oplus \c).$$
By the Gyrotriangle Inequality Axiom, 
$$\Vert \ominus \a\oplus\c\Vert\leq \Vert\ominus\a\oplus \b\Vert\oplus 
\Vert\gyr[\ominus \a,\b](\ominus\b\oplus \c)\Vert=
\Vert\ominus\a\oplus \b\Vert\oplus 
\Vert\ominus\b\oplus \c\Vert.$$
Hence $d_\oplus(\a,\c)\leq d_\oplus(\a,\b)\oplus d_\oplus(\b,\c).$

\end{proof}

\begin{proposition}\label{P:distinv}
The distance $d_\oplus$ is invariant under gyrations and left translations.
\end{proposition}

\begin{proof}
We first note by the invariance of the inner product under gyrations that
$$\Vert \gyr[\a,\b]\u\Vert=(\gyr[\a,\b]\u\cdot\gyr[\a,\b]\u)^{1/2}=(\u\cdot\u)^{1/2}=\Vert\u\Vert.$$
Hence 
\begin{eqnarray*}
d_\oplus(\gyr[\a,\b]\u,\gyr[\a,\b]\v)&=&\Vert \ominus\gyr[\a,\b]\u\oplus\gyr[\a,\b]\v\Vert\\
&=&\Vert\gyr[\a,\b](\ominus \u\oplus\v)\Vert\\
&=&\Vert\ominus \u\oplus\v\Vert=d_\oplus(\u,\v).
\end{eqnarray*}

For the second assertion, we first note by Lemma \ref{gyrotrans1} that
$$\ominus(\a\oplus\u)\oplus(\a\oplus\v)=\gyr[\a,\u](\ominus\u\oplus\v),$$
and hence
\begin{eqnarray*}
d_\oplus(\a\oplus\u,\a\oplus\v)&=&\Vert \ominus(\a\oplus\u)\oplus(\a\oplus\v)\Vert\\
&=&\Vert \gyr[\a,\u](\ominus\u\oplus\v)\Vert\\
&=& \Vert \ominus\u\oplus\v\Vert=d_\oplus(\u,\v).
\end{eqnarray*}

\end{proof}

\subsection{The Einstein gyrovector space inner product}
We have considered previously the gyrogroup $\R_c^3$ of admissible velocities under Einstein velocity
addition.  We have also introduced  the exponential map $\exp:\R^3\to \R_c^3$ defined by
$\exp(\u)=c\tanh(\vert\u\vert)(\u/\vert \u\vert).$  Via the exponential map and its inverse log,
we can define the scalar multiplication by 
$t.\v:=\exp(t\log \v).$  The fact that $\exp$ restricted to the one-dimensional subspaces of $\R^3$
is a homomorphism preserving scalar multiplication yields the axioms of a gyrovector space (see 
Section 6.1).  

We equip $\R_c^3$ with the usual euclidean inner product on $\R^3$ restricted to $\R_c^3$.
We need some preparation to establish the appropriate axioms for the inner product.
We recall from Proposition \ref{P:boostcomp} that 
for $\u,\v\in \R_c^3$ and the corresponding Lorentz boosts $B(\u),B(\v)$, 
$$B(\u)B(\v)=B(\u\oplus\v)h(\u,\v), \mbox{ where }  h(\u,\v)\in O^+(1,3).$$
We note that actually $h(\u,\v)\in SO^+(1,3)$,  since in the preceding equation it must have a positive determinant
for equality to hold, since each Lorentz boost has a positive determinant.

\begin{proposition}\label{P:gyr=}
For $\u,\v\in \R_c^3$, $\gyr[\u,\v]=S(\u,\v)\in SO^+(1,3)$,  where $S(\u,\v)$ is the $3\times 3$-block
matrix in the block diagonal matrix $h(\u,\v)$.
\end{proposition}

\begin{proof}
For $\u,\v,\w\in\R_c^3$, we have
$$B(\u)(B(\v)B(\w)=B(\u)B(\v\oplus \w)h(\v,\w)=B(\u\oplus(\v\oplus\w))h(\u,\v\oplus\w)h(\v,\w)$$
and associating the other way
$$B(\u)B(\v)B(\w)=B(\u\oplus\v)h(\u,\v)B(\w)=B(\u\oplus\v)B(\w)^{h(\u,\v)}h(\u,v),$$
where $B(\w)^{h(\u,\v)}=h(\u,\v)B(\w)h(\u,\v)^{-1}$.  Since $h(\u,\v)\in SO^+(1,3)$, by Proposition \ref{P:orthfact}
it has a block diagonal form with diagonal entries $1$, $S=S(\u,\v)\in SO(3)$.  Hence
\begin{eqnarray*}
B(\w)^{h(\u,\v)}&=& \left[\begin{matrix} 1 & 0\\ 0 & S\end{matrix}\right]\left[\begin{matrix} \gamma & *\\ \frac{\gamma}{c}\w & *\end{matrix}\right]
\left[\begin{matrix} 1 & 0\\ 0 & S^{-1}\end{matrix}\right]\\
&=&\left[\begin{matrix} \gamma & *\\  \frac{\gamma}{c}S\w& *\end{matrix}\right]=B(S\w),
\end{eqnarray*}
where $\gamma=\gamma_\w$ and where the last equality holds since the conjugation 
of $B(\w)$ must again be a Lorentz boost, and the first column
then determines it.  We now continue our earlier computation:
$$B(\u\oplus \v)B(\w)^{h(\v,\v)}h(\u,\v)=B(\u\oplus \v)B(S\w)h(\u,\v)=B((\u\oplus\v)\oplus S\w)h(\u\oplus\v,S\w)h(\u,\v).$$
Comparing with our first computation, we conclude $B(\u\oplus(\v\oplus\w))=B((\u\oplus\v)\oplus S\w)$.  Since $B$ is
bijective, 
$$(\u\oplus \v)\oplus\gyr[\u,\v]\w=\u\oplus(\v\oplus\w)=(\u\oplus\v)\oplus S\w,$$
and  by left cancellation $\gyr[\u,\v]=S=S(\u,\v)$.
\end{proof} 

From Proposition \ref{P:gyr=} we have immediately the
\begin{corollary}\label{C:preservedot}
The gyrations $\gyr[\u,\v]$ of $\R_c^3$ preserve the inner product.
\end{corollary}

\subsection{Rapidities and norm axioms}
The notion of the \emph{rapidity} 
$$\phi_\v=\tanh^{-1}\frac{\Vert\v\Vert}{c}$$
of an admissible velocity $\v$ was introduced very early in the development of special relativity.
Rapidities satisfy a number of useful properties.
\begin{lemma}\label{L:rapid1}
Let $\v\in\R_c^3$ be an admissible velocity vector.
\begin{itemize}
\item[(1)] $c\tanh(\phi_\v)=\Vert\v\Vert$.
\item[(2)] $\v=\exp\frac{\phi_\v\v}{\Vert\v\Vert}$; hence $\log \v=\frac{\phi_\v\v}{\Vert\v\Vert}$.
\item[(3)] $\cosh\phi_\v=\gamma_\v$.
\item[(4)] $\sinh\phi_\v=\gamma_\v\frac{\Vert\v\Vert}{c}$.
\end{itemize}
\end{lemma}

\begin{proof}
Item (1) follows directly from the definition of the rapidity.  Applying the definition
$\exp(\u)=c\tanh(\Vert\u\Vert)(\u/\Vert\u\Vert)$ to $\phi_\v\v/\Vert\v\Vert$ and noting
that the latter's norm is $\phi_\v$, we have 
$$\exp\frac{\phi_\v\v}{\Vert\v\Vert}=c\tanh(\phi_\v)\frac{\v}{\Vert\v\Vert}=\v,$$
where the last equality follows from (1).

For (3), we have that
$$\gamma_\v=(1-\frac{\Vert\v\Vert^2}{c^2})^{-1/2}=(1-\tanh^2 \phi_\v)^{-1/2}=(\mbox{sech}^2\phi_\v)^{-1/2}
=\cosh \phi_\v.$$
For (4), $\sinh \phi_\v=\cosh \phi_\v\tanh\phi_\v=\gamma_\v\frac{\Vert\v\Vert}{c}$.
\end{proof}

The exponential $\exp:\R\to\R_c$ is given by $\exp(x)=c\tanh(x)=c\tanh(\vert x\vert)(x/\vert x\vert)$, where the
last equality holds for $x\ne 0$ and follows from the fact $\tanh(x)$ is an odd function. 
\begin{lemma}\label{L:scalar2}
For $r\in\R$, $\v\in \R_c^3$, $\v\ne 0$, $r.\v=(r.\Vert\v\Vert)(\v/\Vert\v\Vert)$.
\end{lemma}

\begin{proof}
Applying our previous results, we obtain for $r\ne 0$,
\begin{eqnarray*}
r.\v&=& \exp(r\log \v)=\exp\bigg(\frac{r\phi_\v\v}{\Vert\v\Vert}\bigg)\\
&=& c\tanh(\vert r\vert\phi_\v)\frac{r\phi_\v\v}{\vert r\vert\phi_\v\Vert\v\Vert}\\
&=& \exp(r\phi_\v)\frac{\v}{\Vert \v\Vert}.
\end{eqnarray*}
By Lemma \ref{L:rapid1}(1)  $\Vert\v\Vert=c\tanh(\phi_\v)=\exp\phi_\v$, so
$\log\Vert \v\Vert=\phi_\v$.  Applying this to the previous equalities
yields $$r.\v=\exp(r\phi_\v)\frac{\v}{\Vert \v\Vert}=\exp(r\log \Vert\v\Vert)\frac{\v}{\Vert \v\Vert}
=(r.\Vert\v\Vert)\frac{\v}{\Vert \v\Vert}.$$
The case $r=0$ is trivial.
\end{proof}

We now verify further axioms of an inner product gyrovector space.
\begin{lemma}\label{L:more}
For $r\in\R$, $\v\in R_c^3$, $\Vert r.\v\Vert=\vert r\vert.\Vert\v\Vert$.
\end{lemma}

\begin{proof}
Equality trivially holds for the cases $r=0$ or $\v=0$.  So we assume both
are not $0$. By the preceding lemma
$$\Vert r.\v\Vert=\Vert (r.\Vert\v\Vert)\frac{\v}{\Vert\v\Vert}\Vert=\vert (r.\Vert\v\Vert)\vert.$$
If $r>0$, then $\vert r.\Vert\v\Vert\,\vert=\vert r\vert.\Vert\v\Vert$ and we are done.  
If $r<0$, then $r.\Vert\v\Vert<0$, so $\vert r.\Vert\v\Vert\,\vert=-r.\Vert\v\Vert
=\vert r\vert.\Vert\v\Vert$.
\end{proof}

\begin{lemma}
For $0\ne r\in\R$ and $0\ne\v\in \R_c^3$,  $\frac{\vert r \vert.\v}{\Vert r.\v\Vert}=
\frac{\v}{\Vert\v\Vert}$.
\end{lemma}

\begin{proof}
By Lemmas \ref{L:scalar2} and \ref{L:more}
$$\frac{\vert r \vert.\v}{\Vert r.\v\Vert}=\frac{(\vert r\vert.\Vert\v\Vert)\v}{(\vert r\vert.\Vert\v\Vert )\Vert\v\Vert }
=\frac{\v}{\Vert\v\Vert}.$$
\end{proof}

The final axiom that we need to verify is the triangular inequality.
\begin{proposition}\label{P:trineq}
For $\u,\v\in \R_c^3$, $\Vert \u\oplus\v\Vert\leq \Vert\u\Vert\oplus\Vert\v\Vert$.
\end{proposition}

\begin{proof}
From Problem \ref{gammaform}, equation \ref{E:gammaplus}, we have
\begin{eqnarray*}
\gamma_{\Vert \u\Vert\oplus\Vert\v\Vert}&=&\gamma_\u \gamma_\v\Big(1+\frac{\Vert\u\Vert\,\Vert\v\Vert}{c^2}\Big)\\
&\geq & \gamma_\u\gamma_\v\Big(1+\frac{\u\cdot\v}{c^2}\Big)\\
&=& \gamma_{\u\oplus\v}=\gamma_{\Vert\u\oplus\v\Vert}.
\end{eqnarray*}
Since $\gamma_\x=\gamma_{\Vert\x\Vert}$ is a monotonically increasing function
of $\Vert\x\Vert$, it follows that
$$\Vert\u\oplus\v\Vert\leq \Vert\u\Vert\oplus\Vert\v\Vert.$$

\end{proof}

We have thus shown 
\begin{theorem}
The Einstein gyrovector space $(\R_c^3,\oplus,.)$ is a real inner product gyrovector space.
\end{theorem}

\section{A Little Hyperbolic Geometry}
We have seen how to define distance and length in a real inner product gyrovector space, although it might better
be called a ``gyrolength" since it takes values not in the nonnegative reals, but in the nonnegative members
of $(\R_c,\oplus,.)$.  In this section we consider briefly how to extend other basic aspects of vector analysis in euclidean
spaces to the hyperbolic setting of real inner product gyrovector spaces.

In addition to lengths we can also measure angles from the formula
$$\cos\alpha:=\frac{\u}{\Vert\u\Vert}\cdot\frac{\v}{\Vert\v\Vert},$$
where $\alpha$ is the measure of the angle at $0$ between the vectors $\u$ and $\v$.  
Note that the preceding equation can be rewritten in the familiar form
$$\u\cdot \v=\Vert\u\Vert \Vert \v\Vert \cos\alpha.$$

More generally, if $A,B,C$ are noncollinear points (i.e., don't lie on a gyroline), we calculate the 
measure of the gyroangle $\angle ABC$ determined by
the rooted gyrovectors $\oa{BA}$ and $\oa{BC}$ from the formula
$$\cos \alpha:=\frac{\oa{BA}}{\Vert\oa{BA}\Vert}\cdot \frac{\oa{BC}}{\Vert\oa{BC}\Vert}=\frac{\ominus B\oplus A}
{\Vert \ominus B\oplus A\Vert}\cdot \frac{\ominus B\oplus C}{\Vert \ominus B\oplus C\Vert}.$$

\begin{definition}  We say two segments of gyrolines are \emph{congruent} if the endpoints $P,Q$  of the first segment are
the same distance apart as those $P',Q'$of the second segment, i.e., $d_\oplus (P,Q)=d_\oplus(P',Q')$.  We write
$\overline{PQ}\cong\overline{P'Q'}$.  We say two gyroangles $\angle BAC$ and $\angle B'A'C'$ are \emph{congruent}
if their measures are equal and write $\angle BAC\cong\angle B'A'C'$.  Two triangles are \emph{congruent} if  all corresponding
sides and gyroangles are congruent.
\end{definition} 
\begin{proposition}\label{P:gyrangle} Let $A,B,C$ be noncollinear points.
\begin{itemize}
\item[(i)]  If $A'=\gyr[\u,\v] A$, $B'=\gyr[\u,\v]B$, and $C'=\gyr[\u,\v]C$, then 
$\angle ABC\cong\angle A'B'C'$, i.e., the 
measurement of the gyroangles are equal.
\item[(ii)]  If $A'=P\oplus A$, $B'=P\oplus B$, and $C'=P\oplus C$, then
$\angle ABC\cong\angle A'B'C'$.
\end{itemize}
\end{proposition}

\begin{problem} Prove Proposition \ref{P:gyrangle}. (Hint: Modify
the proof of Proposition \ref{P:distinv}.)
\end{problem}

We turn now to the consideration of triangles.  Let $A,B,C$ be noncollinear points and the vertices of a triangle.  
We then have gyroangles at $A,B,C$ denoted $\angle A$, $\angle B$, and $\angle C$ resp.  We orient the sides
of $\Delta ABC$ as rooted gyrovectors $\oa{AB}$, $\oa{CB}$, and $\oa{CA}$.  We set $\a=\ominus C\oplus B$,
$\b=\ominus C\oplus A$, and $\c=\ominus A\oplus B$, the sides opposite $\angle A$, $\angle B$, and 
$\angle C$ resp.  We denote their lengths by 
\begin{eqnarray*}
a&=&\Vert\a\Vert=\Vert\oa{CB}\Vert=\Vert \ominus C\oplus B\Vert=d_\oplus(C,B)=d_\oplus(B,C)\\
b&=&\Vert\b\Vert=\Vert\oa{CA}\Vert=\Vert\ominus C\oplus A\Vert=d_\oplus(C,A)=d_\oplus(A,C)\\
c&=&\Vert\c\Vert=\Vert\oa{AB}\Vert=\Vert\ominus A\oplus B\Vert=d_\oplus(A,B)=d_\oplus(B,A).
\end{eqnarray*}
In particular, the lengths $a,b,c$ are independent of the orientation chosen.
\begin{problem}\label{tri=}
Show that for the given orientation
$$\Vert \ominus \a\oplus \b\Vert=\Vert \ominus(\ominus C\oplus B)\oplus(\ominus C\oplus A)\Vert=
\Vert \ominus B\oplus A\vert=\Vert \c\Vert.$$ 
\end{problem}

\begin{problem}  Show that \begin{equation}\label{E:cos2}
\cos\angle C=\frac{\a\cdot\b}{ab},
\end{equation} in full analogy to the euclidean case.
\end{problem}

\subsection{Relativistic hyperbolic geometry}
In this section we work in the Einstein gyrovector space $\R_s^3$, the open ball of radius $s$, where we switch from
$c$ to $s$ to avoid notational confusion.
We recall equation \ref{E:gammaplus}, the gamma identity,
\begin{equation}\label{E:gammaplus2}
\gamma_{\u\oplus\v}=\gamma_\u \gamma_\v\Big(1+\frac{\u\cdot\v}{s^2}\Big).
\end{equation}
from Problem \ref{gammaform}.
Since $\gamma_\v=\gamma_{\Vert \v\Vert}$, we have from Problem \ref{tri=} that
$$\gamma_\c=\gamma_{\ominus\a\oplus\b}$$
 in triangle $\Delta ABC$ of the preceding section.
It follows that 
$$\gamma_c=\gamma_\c=\gamma_{\ominus\a\oplus\b}=\gamma_{\ominus \a}\gamma_\b\Big(1+\frac{\ominus \a\cdot\b}{s^2}\Big)
=\gamma_a\gamma_b\Big(1-\frac{ab\cos \angle C}{s^2}\Big),$$
since $\ominus \a=-\a$.  
\begin{problem}  Show that $\ominus \a=-\a$ in the Einstein gyrovector space.
\end{problem}
We slightly change the notation and record the preceding equation as the \emph{relativistic law of cosines}.
\begin{proposition}\label{P:lawcos}
In $\Delta ABC$, we have 
\begin{equation}\label{E:cos1}
\gamma_a=\gamma_b\gamma_c\Big(1-\frac{bc}{s^2}\cos\angle A\Big),
\end{equation}
where $a,b,c$ are the lengths of the sides opposite $\angle A$, $\angle B$, $\angle C$, resp.
\end{proposition}

\begin{problem}
Show that 
$$\frac{a^2}{s^2}=\frac{\gamma_a^2-1}{\gamma_a^2}.$$
\end{problem}

\begin{problem}\label{cos=3}
Use the preceding problem and equation (\ref{E:cos1}) to show
$$\cos\angle A=\frac{\gamma_b\gamma_c-\gamma_a}{\sqrt{\gamma_b^2-1}\sqrt{\gamma_c^2-1}}.$$
\end{problem}
Note the right-hand side of the preceding equation allows one to calculate $\cos\angle A$ from the lengths $a,b,c$ of the sides
of $\Delta ABC$.  Thus the radian measure of  $\angle A$ is uniquely determined in $(0,\pi)$.
We thus obtain
\begin{theorem}(SSS) Two triangles are congruent if their corresponding sides are congruent.
\end{theorem}

\begin{problem}
Use the preceding theory to deduce the SAS theorem in the geometry of $\R_s^3$.
\end{problem}

For a gyroangle $\angle A$ define $\sin\angle A:=(1-\cos^2\angle A)^{1/2}$. One can 
establish the equality
\begin{equation}\label{gamma=}
\gamma_a=\frac{\cos\angle A+(\cos\angle B)\cos\angle C)}{\sin\angle B\sin\angle C}
\end{equation}
by direct computation by using the equation of Problem \ref{cos=3} to establish the 
variant form
$$\gamma_a^2=\frac{(\cos\angle A+\cos\angle B\cos\angle C)^2}{(1-\cos^2\angle B)(1-\cos^2\angle C)}.$$
One can also derive the  \emph{relativistic law of sines}:
\begin{equation}
\frac{\sin\angle A}{\gamma_a a}=\frac{\sin\angle B}{\gamma_b b}=\frac{\sin\angle C}{\gamma_c c}.
\end{equation}
\begin{problem}
Derive one of the two preceding equations.  (If you derive the second, you may assume the first.)
\end{problem}

\begin{problem}  Show that an equilateral triangle with $\angle A=\angle B=\angle C=\theta$ has sides with 
length $\sqrt{2\cos\theta-1}/\cos\theta$.  (Be aware that $\theta<60^\circ$ in the hyperbolic setting.)
\end{problem}
Note that the angle determines the side in this setting, which is certainly not the case in euclidean geometry.


\begin{thebibliography}{99}

\bibitem{Fri}
Y.\ Friedman, Physical Applications of Homogeneous Balls, Birkh\"auser, 2005.

\bibitem{KL}
S.\ Kim and J.\ Lawson, Smooth Bruck loops, symmetric spaces, and nonassociative vector spaces,
Demonstratio Math. \textbf{44} (2011), 755-779.
 
\bibitem{Ung} A. A. Ungar, Analytic Hyperbolic Geometry and Albert Einstein's Special Theory of Relativity, World 
Scientific Press, 2008.

\end{thebibliography}
\end{document}